\documentclass[12pt]{iopart}

\usepackage{iopams}
\expandafter\let\csname equation*\endcsname\relax
\expandafter\let\csname endequation*\endcsname\relax

\usepackage{amsfonts,amsmath,amssymb}
\usepackage[english]{babel}
\usepackage[normalem]{ulem}

\usepackage{amsthm}
\usepackage{graphicx}
\usepackage{dsfont}
\usepackage[caption=false]{subfig}
\usepackage{epsfig}
\usepackage{ulem}
\usepackage{tikz}
\usepackage{mathrsfs}
\usepackage{url}

\usepackage[utf8]{inputenc}
\usepackage{placeins}

\usepackage{pgfplots}
\usepackage{pgffor}
\usetikzlibrary{plotmarks}
\pgfplotsset{compat=1.8}

\usepackage{bbold}             
\usepackage{bm}
\usetikzlibrary{shapes,arrows}
\usetikzlibrary{matrix}
\usetikzlibrary{calc,patterns,angles,quotes,arrows,automata}
\graphicspath{ {./imgs/}}

\usepackage[linesnumbered,noend]{algorithm2e}
\usepackage{algorithmic}
\RestyleAlgo{ruled}
\SetKwInOut{Input}{Input}
\SetKwInOut{Output}{Output}
\SetKwInOut{Parameters}{Parameters}
\SetKwIF{If}{ElseIf}{Else}{if}{:}{else if}{else}{endif}

\theoremstyle{definition}

\newtheorem{lemma}{Lemma}
\newtheorem{proposition}{Proposition}
\newtheorem*{proposition*}{Proposition}

\newtheorem{theorem}{Theorem}
\newtheorem{corollary}{Corollary}
\newtheorem{definition}{Definition}

\newcommand{\one}{\mathds{1}}
\newcommand{\ket}[1]{\lvert #1 \rangle}
\newcommand{\bra}[1]{\langle #1 \rvert}
\newcommand{\braket}[2]{\langle #1 \lvert #2 \rangle}
\newcommand{\ketbra}[2]{\lvert #1 \rangle \langle #2 \rvert}
\newcommand{\iu}{\mathrm{i}\mkern1mu}

\definecolor{color1}{rgb}{0.0, 0.6056031704619725, 0.9786801190138923}
\definecolor{color2}{rgb}{0.8888735440600661, 0.435649148506399, 0.2781230452972764}
\definecolor{color3}{rgb}{0.24222393333911896, 0.6432750821113586, 0.304448664188385}
\definecolor{color4}{rgb}{0.7644400000572205, 0.4441118538379669, 0.8242975473403931}
\definecolor{color5}{rgb}{0.6755439043045044, 0.5556622743606567, 0.09423444420099258}

\begin{document}

\title{Quantum Wall States for Noise Mitigation and Eternal Purity Bounds}

\author{Miguel Casanova$^{1,*}$, Francesco Ticozzi$^{1,2}$}

\address{$^1$ Department of Information Engineering, University of Padova, Italy}
\address{$^2$ QTech center, University of Padova, Italy}
\address{$^*$ Corresponding author}
\eads{\mailto{casanovame@dei.unipd.it}$\dagger$ and \mailto{ticozzi@dei.unipd.it}}

\begin{abstract}
The present work analyzes state-stabilization techniques for decoupling a subsystem from environmental interactions.
The proposed framework uses analytical and numerical tools to find an approximate decoherence-free subspace (DFS)
with enhanced passive noise isolation.
Active state-stabilizing control on a subsystem mediating dominant environmental interactions, which we call wall subsystem,
creates an effective quantum wall state.
The proposed method controls only the wall subsystem, leaving the logical subsystem untouched.
This simplifies logic operations in the protected subsystem, and makes it suitable for integration
with other quantum information protection techniques, such as dynamical decoupling (DD).
We demonstrated its effectiveness in enhancing the performance of selective or complete DD.
Under suitable conditions, our method maintains system purity above a threshold for all times,
achieving eternal purity preservation. Theoretical analysis links this behavior to the asymptotic
spectrum of the Hamiltonian when the control gain grows unbounded.
\end{abstract}


\submitto{Quantum Science and Technology}

\maketitle

\section{Introduction}

In the effort of creating quantum information processing devices able to move the technology past the NISQ era \cite{preskillQuantumComputingNISQ2018}, for many potential physical implementations the detrimental effect of noise still represents a crucial barrier. To avoid or mitigate this issue, several noise avoidance and noise suppression techniques have been proposed, supplementing efforts to engineer suitable physical layers.
The first class, which focuses on passive protection from errors by encoding information in decoherence-free subspaces and noiseless subsystems \cite{lidarDecoherenceFreeSubspacesQuantum1998, zanardiNoiselessQuantumCodes1997, benattiIrreversibleQuantumDynamics2003}, offers the advantage of having the manipulation of the encoded information unhindered by the presence of control sequences, or correction protocols. However, the presence of a viable noiseless code is granted only in presence of symmetries in the (dominant) coupling between the system and the environment, and the information might need to be stored in some states that are highly entangled with respect to the physical subsystems decomposition, which might thus be hard to prepare.
The noise suppression methods, such as dynamical decoupling \cite{violaDynamicalDecouplingOpen1999} and active error correction \cite{knillTheoryQuantumErrorcorrecting1997, knillTheoryQuantumError2000}, (see \cite{brunQuantumErrorCorrection2020} and references therein for a more comprehensive review) target the effect of noise effects via symmetrization or undoing the information leakage from a given code by employing typically fast and strong control actions, which have to be intertwined with the control actions needed to implement the desired information processing \cite{KodjastehDynamically2009}.
 
In this work, we offer a pathway toward the construction approximately noiseless codes by the use of active methods, where the controls limiting the effect of noise do not act directly on the encoded information. In this sense, the proposed method can be seen as a {\em hybrid} noise protection approach.
We proceed by first separating the degrees of freedom of the system of interest in two factor (virtual) subsystems: the {\em logical subsystem}, corresponding to the support of the encoded information, and the {\em wall subsystem}, mediating the leading interaction terms with the environment. The state and dynamics of the latter are then engineered to block such interactions by preparing and stabilizing a state that minimizes purity loss in the code subsystem, which we call the {\em wall state}.

The interaction between the system and its environment is assumed to be Hamiltonian, while the dynamics on the latter is allowed to include dissipative, time-homogenous dynamics in Lindblad form. That is, we consider a {\em structured} environment, which combines non-Markovian coupling with the system and is allowed to further interact with a memory-less reservoir, leading to Markovian dissipation terms.
The toolbox we develop includes: 
\begin{enumerate}
\item A numerical method, based on gradient-based Rimeannian optimization, for the optimal selection of the code and wall subsystem;

\item  A prescription for the selection of a special state of the interface subsystem, the {\em wall state}, to minimize purity loss on the code. This can be done either numerically or following an heuristic strategy, which is proved to be optimal in the qubit case;

\item  Multiple ways to actively stabilize the wall state and improve the lifetime of encoded information. We consider engineered Markovian dynamics in Lindblad form, zeno-inducing repeated measurements, and strong Hamiltonian driving.
\end{enumerate}

In building our methodologies, we also show that points 1. and 2. above can be interpreted as a way to find and stabilize approximate decoherence-free subspace, see \cite{ticozziQuantumInformationEncoding2010,wangNumericalMethodFinding2013, hamannApproximateDecoherenceFree2022} and references therein.
From another viewpoint, the method can be seen as an alternative way to enforce a selective quantum Dynamical Decoupling (DD) \cite{violaDynamicalDecouplingOpen1999,ticozziDynamicalDecouplingQuantum2006} of the code subsystem. Two differences with typical dynamical decoupling scenarios are to be remarked: first, we do not aim to decouple the whole physical system of interest, but only a portion of it, whose dimension is fixed by the QIP task at hand; second, we are only aiming to suppress the leading term of the interaction, similar to what is done e.g. in \cite{ticozziDynamicalDecouplingQuantum2006,maGeneralMethodologyDecoupling2011}.
Given this limitation, in the simulation section of the paper we investigate the interplay between our method and standard DD, revealing how the performance of the latter can be improved by combining with our wall-state techniques if some resonance conditions are avoided. A third connection emerges with the approach developed in \cite{khodjastehPointerStatesEngineered2011}, where pulsed sequences are used in non-Markovian open systems to effectively generate approximate fixed points, or {\em pointer states}. Our problem of interest can be seen as a generalization of the one considered there, where the target to be stabilized is not a single state but a whole subspace, or subsystem; the methods we use are also different, as we couple passive encodings with strong coherent or dissipative driving. 

In addition, we show how stabilizing the wall state by strong Hamiltonian driving can lead to a remarkable feature, not exhibited by the other incoherent control methods: the purity of the code system can be preserved above a certain threshold for arbitrarily long times,  giving rise to an effective {\em eternal purity bounds}. An explanation of the phenomenon in terms of the evolution of the coherences is provided, together with sufficient conditions for its emergence based on spectral gaps in the underlying Hamiltonian dynamics.

Some seminal ideas used in this work have been presented in \cite{casanovaStabilizingWallStates2024}, in a simplified setting where the logical and interface (wall) subsystems were assumed to be known and many technical details and proofs were missing. We here generalize, advance, and provide formal proofs for the results of that work, developing a complete theory and a {\em systematic toolobox} to wall-state engineering. In particular (each item highlights a key improvement with respect to the previous work), we: 
\begin{itemize}
    \item[(i)] Pose the problem for a general finite-dimensional controlled system coupled to an environment, possibly structured and/or with Markovian dissipation, via a general interaction Hamiltonian (Section \ref{sec:setup}; before we only considered a  logical-wall subdivision of the system to be given);  
    \item[(ii)] Compare and interpret the approach in the framework of decoherence-free subspaces (DFS), showing that {\em perfect wall states} are equivalent to DFSs for systems composed by $n$ identical subsystems (Section \ref{sec:dfs}): in this sense, our general method can be seen as a way to find the best approximate DFS one can engineer in a given system;
    \item[(iii)] Develop a method for {\em finding} optimized logical-wall subsystem decomposition, including a critical regularization functional (Section \ref{sec:subsystem}); 
    \item[(iv)] Analyze in detail the purity dynamics, and propose optimization problems to select the wall states that minimizes the initial purity acceleration (Section \ref{sec:state}), with an analytical solution for the qubit case; 
    \item[(v)] Demonstrate the performance of our wall state selection algorithms and different stabilization techniques in richer and more complex examples, including a {\em 5 spin-1/2 chain} with three-body interactions and a {\em central spin model} (Section \ref{sec:examples}); 
    \item[(vi)] Compare and integrate the method with dynamical-decoupling techniques (Section \ref{sec:dd});
\item[(vii)] Provide an in-depth  analysis and spectral conditions that guarantee the emergence of {\em eternal purity preservation}, a remarkable phenomenon one can observe, under suitable conditions, while stabilizing the wall state with a strong Hamiltonian driving (Section \ref{sec:eternal}).
\end{itemize}

\section{Setup}\label{sec:setup}

Consider a finite-dimensional quantum system with Hilbert space
${\cal H} = {\cal H}_\mathfrak{s} \otimes {\cal H}_\mathfrak{e}$, where ${\cal H}_\mathfrak{s}$ corresponds to
a controllable system of dimension $n_\mathfrak{s},$ and ${\cal H}_\mathfrak{e}$ to the environment of dimension
$n_\mathfrak{e}$.
We consider a decomposition of the system's Hamiltonian into local and interaction terms, i.e.
\begin{equation}
    H = H_\mathfrak{s} + H_\mathfrak{e} + H_\mathfrak{se},
\end{equation}
where $H_\mathfrak{s} = \check{H}_\mathfrak{s} \otimes \one_\mathfrak{e}$
and $H_\mathfrak{s} = \one_\mathfrak{s} \otimes \check{H}_\mathfrak{e}$.
In addition, we consider cases in which the evolution of the environment subsystem
includes dissipative effects described by a Markovian master equation \cite{goriniCompletelyPositiveDynamical1976, lindbladGeneratorsQuantumDynamical1976}. Therefore, the
state dynamics follows
\begin{equation}
\dot\rho(t) = {\cal L}(\rho) = - \iu [H, \rho(t)] + \sum_m {\cal D}_{L_m}(\rho),
\end{equation}
where ${\cal D}_{L_m}(\rho) = L_m \rho(t) L_m^\dagger - \frac{1}{2} \{L_m^\dagger L_m, \rho(t)\},$ where
 we only consider Lindblad operators of the following form,
\begin{equation}
L_m = \one_\mathfrak{s} \otimes L_m^\mathfrak{e}.
\end{equation}
This compactly models a structured environment $\cal H_\mathfrak{e} \otimes H_\mathfrak{m}$, where we traced out $\mathcal{H}_\mathfrak{m}$,
an infinite-dimensional subsystem that
induces Markovian dynamics in $\mathcal{H}_\mathfrak{e}$ and has no interaction with $\mathcal{H}_\mathfrak{s}$. 
Under these assumptions, the evolution of $\rho_\mathfrak{s}=\tr_\mathfrak{e}(\rho)$ becomes non unitary only due to its Hamiltonian coupling with the environment via $H_\mathfrak{se}.$

The first task of our work is to identify the subsystem of ${\cal H}_\mathfrak{s}$ subject to the dominant
terms of the interaction with ${\cal H}_\mathfrak{e}$. Once this subsystem has been identified,
we can decompose ${\cal H}_\mathfrak{s}$ into such a subsystem ${\cal H}_\mathfrak{w},$ which we shall call the {\em wall subsystem},
and its tensor complement $\cal H_\mathfrak{l}$ the {\em logical subsystem}, defined such that $\cal H_\mathfrak{s} = H_\mathfrak{l} \otimes H_\mathfrak{w}$.
This decomposition allows us to apply arbitrary control actions on ${\cal H}_\mathfrak{w}$,
even dissipative ones,  without directly affecting the information stored in the logical subsystem. We shall look for control actions such that the purity of the reduced state in ${\cal H}_\mathfrak{l}$ is maximized.

This setup allows for the implementation of quantum gates or other quantum information
processing protocols to ${\cal H}_\mathfrak{l}$ completely independently of the noise-suppressing
control protocol, which is applied to ${\cal H}_\mathfrak{w}$ instead, in what could be called a {{\em decoupled decoupling}} strategy.
In this way, we ``sacrifice'' the possibility of encoding quantum information 
in ${\cal H}_\mathfrak{w}$ and instead transform its state as a control resource. A state optimized for noise protection will be
called a {\em wall state.} 

\section{Motivating example: Ising chain}

Consider as a toy model the one-dimensional Ising model of length $N$ with longitudinal field.
Each of the spins in the chain has an associated Hilbert space $\mathcal{H}_i \simeq \mathbb{C}^2, i = 1, ..., N$,
whereas the whole chain has a Hilbert space ${\cal H} = \bigotimes_{i=1}^N {\cal H}_i$.
Suppose that we aim to use the spins up to the $j$-th to store quantum information that we
want to preserve, whereas the spins starting from $j+1$ are uncontrollable.
Therefore, we take the latter to act as
the environment, i.e. ${\cal H}_\mathfrak{s} = \bigotimes_{i=1}^j {\cal H}_i$,
${\cal H}_\mathfrak{e} = \bigotimes_{i=j+1}^N {\cal H}_i$ and $\cal H = H_\mathfrak{s} \otimes H_\mathfrak{e}$.
Let $\sigma_{x,y,z}$ denote the Pauli matrices, and define
\begin{equation}
\label{eq:Jz}
J_i^z = \one^{\otimes i-1} \otimes \frac{1}{2}\sigma_z \otimes \one^{\otimes N-i}.
\end{equation}
The system has the following Hamiltonian,
\begin{equation}
H = h \sum_{i=1}^N J_i^z + g_z \sum_{i=1}^{N-1} J_i^z J_{i+1}^z,
\end{equation}
which can then be decomposed into the following local and interaction terms
\begin{equation}
H_\mathfrak{s} = h \sum_{i=1}^j J_i^z + g_z \sum_{i=1}^{j-1} J_i^z J_{i+1}^z,
\end{equation}
\begin{equation}
H_\mathfrak{e} = h \sum_{i=j+1}^N J_i^z + g_z \sum_{i=j+1}^{N-1} J_i^z J_{i+1}^z,
\end{equation}
\begin{equation}
H_\mathfrak{se} = g_z J_j^z J_{j+1}^z.
\end{equation}

It is easy to see that for this model it suffices to use the $j$-th spin as wall subsystem in order to
decouple the two sides of the chain. More precisely, as above,
${\cal H}_\mathfrak{s} = {\cal H}_\mathfrak{l} \otimes {\cal H}_\mathfrak{w}$,
where ${\cal H}_\mathfrak{l} = \bigotimes_{i=1}^{j-1} {\cal H}_i$ and
${\cal H}_\mathfrak{w} = {\cal H}_j$. Under this setup, ${\cal H}_\mathfrak{l}$ may
be perfectly decoupled from the environment just by setting the initial state of
${\cal H}_\mathfrak{w}$ equal to either $\ket{0}_j$ or $\ket{1}_j$, where
$J_j^z \ket{0}_j = \ket{0}_j$ and $J_j^z \ket{1} = - \ket{1}$.
Let us fix $\ket{w} = \ket{0}_j$, then we can verify that if
$\rho(0) = \rho_\mathfrak{l} \otimes \ketbra{w}{w} \otimes \rho_\mathfrak{e}$, we have that
\begin{equation}
\rho(t) = e^{-\iu H t} \rho(0) e^{\iu H t} =
\tilde\rho_\mathfrak{l}(t) \otimes \ketbra{w}{w} \otimes \tilde\rho_\mathfrak{e}(t),
\end{equation}
where $\tilde\rho_\mathfrak{l} = U_\mathfrak{l}(t) \rho_\mathfrak{l} U_\mathfrak{l}(t)$,
$\tilde\rho_\mathfrak{e} = U_\mathfrak{e}(t) \rho_\mathfrak{e} U_\mathfrak{e}(t)$ and in turn
\begin{equation}
U_\mathfrak{l}(t) = (\prod_{i}^{j-1} e^{-\iu h J_i^z t}) (\prod_{i}^{j-2} e^{-\iu g_z J_i^z J_{i+1}^z})
e^{-\iu g_z J_{j-1}^z t}
\end{equation} and
\begin{equation}
U_\mathfrak{e}(t) = (\prod_{i=j+1}^{N} e^{-\iu h J_i^z t}) (\prod_{i=j+1}^{N-1} e^{-\iu g_z J_i^z J_{i+1}^z})
e^{-\iu g_z J_{j+1}^z t}.
\end{equation}

In regards to the evolution of $\ketbra{w}{w}$, we have that it is static, since
$U_\mathfrak{w}(t) \ketbra{w}{w} U_\mathfrak{w}(t)^\dagger = \ketbra{w}{w}$, where
\begin{equation}
U_\mathfrak{w}(t) = e^{-\iu (h + 2 g_z) J_j^z t}.
\end{equation}

It is clear that with this choice of initial state of the $j$-th spin, no correlations can occur
between any of the $i$-th and $k$-th spins, where $i < j$ and $k > j$, since it would
need to happen via the $j$-th spin, which experiences no dynamics. 

\section{Ideal Case: Perfect Decoupling via Wall States}\label{sec:dfs}
\subsection{Perfect wall states}
In the previous example, we saw how
with an appropriate choice of wall subsystem and wall state, one may perfectly
decouple the two subsystems without the need of any kind of additional control. In particular,
no control action is applied to the system where the information is stored, leaving it
fully available to other kinds of information processing controls.
This behavior motivates the following definition.

\begin{definition}
\label{def:perfectwall}
    Consider a system with Hilbert space $\cal H = H_\mathfrak{l} \otimes H_\mathfrak{w} \otimes H_\mathfrak{e}$
    and Hamiltonian $H$. Let the $\cal H_\mathfrak{w}$ partition be initialized to pure state $\ket{w}$
    and the total initial state be factorized $\rho_0 = \rho_\mathfrak{l} \otimes \ketbra{w}{w} \otimes \rho_\mathfrak{e}$.
    Then $\ket{w}$ is said to be a {\bf perfect wall state} if the reduced dynamics on the logical subsystem $\rho_\mathfrak{l}(t) = \tr_\mathfrak{we}(e^{-\iu H t} \rho_0 e^{\iu H t})$ are
    unitary for all initial $\rho_\mathfrak{l}$ and $\rho_\mathfrak{e}$.
\end{definition}

In the Ising chain example the choice of wall subsystem and state is trivial, while in general this might not be the case. Assuming for now that we consider systems such that
all the interactions between the logical subsystem and the environment are mediated by a subsystem, the behavior of the example can be generalized as follows.

\begin{proposition}[Perfect wall eigenstates]
\label{thm:perfectwall}
Consider a system with Hilbert space $\cal H = H_\mathfrak{l} \otimes H_\mathfrak{w} \otimes H_\mathfrak{e}$ with
Hamiltonian $$H = H_\mathfrak{l} + H_\mathfrak{w} + H_\mathfrak{e} + H_\mathfrak{lw} + H_\mathfrak{we},$$ where
$H_\mathfrak{l} =  H_L \otimes \one_\mathfrak{w} \otimes \one_\mathfrak{e}$,
$H_\mathfrak{e} =  \one_\mathfrak{l} \otimes \one_\mathfrak{w} \otimes H_E$,
$H_\mathfrak{w} =  \one_\mathfrak{l} \otimes H_W \otimes \one_\mathfrak{e}$,
$H_\mathfrak{lw} = \sum_j B_j \otimes C_j \otimes \one_\mathfrak{e}$ and
$H_\mathfrak{we} = \sum_j \one_\mathfrak{l} \otimes D_j \otimes E_j$. If there exists a representation such that  all of the $H_W$, $C_j$ and $D_j$ operators share a common eiegenstate, such state is a perfect wall state.
\end{proposition}
\begin{proof} 
Denote the common eigenstate of the $H_W$, $C_j$ and $D_j$ operators as $\ket{w}$.
Then, each of these operators has an eigen-decomposition of the following form:
\begin{equation}
X_k = \lambda_{w,k} \ketbra{w}{w} + \sum_{i \neq w} \lambda_{i,k} \Pi_{i,k} =
\lambda_{w,k} \ketbra{w}{w} + \tilde{H}^{w^\perp}_{\mathfrak{w},k},
\end{equation}
where $X_k\in\{\one_\mathfrak{w}, H_W, C_j, D_j, \forall j\}$, and
$\tilde{H}^{w^\perp}_{\mathfrak{w},k} \ket{w} = \Pi_{i,k} \ket{w} = 0, \forall i \neq w$, and
$\Pi_{i,k} \Pi_{j,k} = 0, \forall i \neq j$.
By substituting this decomposition back into the full Hamiltonian, we obtain
\begin{equation}
H = \sum_i \lambda_{w,i} \tilde{H}_{\mathfrak{l},i} \otimes \ketbra{w}{w} \otimes \tilde{H}_{\mathfrak{e},i} +
\sum_k \tilde{H}_{\mathfrak{l},k} \otimes \tilde{H}^{w^\perp}_{\mathfrak{w},k} \otimes \tilde{H}_{\mathfrak{e},k},
\end{equation}
by collecting all the $\ketbra{w}{w}$ terms together, and where
$\tilde{H}_{\mathfrak{l},k} \in \{\one_\mathfrak{l}, H_L, B_j, \forall j\}$ and
$\tilde{H}_{\mathfrak{e},k} \in \{\one_\mathfrak{e}, H_E, E_j,\forall j\}$.
Since $\tilde{H}^{w^\perp}_{\mathfrak{w},k} \ketbra{w}{w} = 0, \forall k$,
the two sums commute and we can write the evolution operator as
\begin{equation}
U = e^{-\iu H t} = U_w U_{w^\perp},
\end{equation}
where $U_w = e^{-\iu (\sum_i \lambda_{w,i} \tilde{H}_{\mathfrak{l},i} \otimes \ketbra{w}{w} \otimes \tilde{H}_{\mathfrak{e},i}) t}$ and
$U_{w^\perp} = e^{-\iu (\tilde{H}_{\mathfrak{l},k} \otimes \tilde{H}^{w^\perp}_{\mathfrak{w},k} \otimes \tilde{H}_{\mathfrak{e},k}) t}$.
Then, it is easy to verify that
\begin{equation}
U_{w^\perp} \sum_{ik} \alpha_{ik} \ket{i} \otimes \ket{w} \otimes \ket{k} = \sum_{ik} \beta_{ik} \ket{i} \otimes \ket{w} \otimes \ket{k},
\end{equation}
\begin{equation}
U_{w} \sum_{ik} \alpha_{ik} \ket{i} \otimes \ket{w} \otimes \ket{k} = \sum_{ik} \mu_{ik} \ket{i} \otimes \ket{w} \otimes \ket{k}.
\end{equation}
Therefore, the state $\ket{w}$ remains fixed. Finally, since there is no direct interaction between
$\cal H_\mathfrak{l}$ and $\cal H_\mathfrak{e}$, no entanglement can occur between these two subsystems
if the initial state is a product state such that $\tr_\mathfrak{le}(\rho) = \ketbra{w}{w}$.
\end{proof}

When a perfect wall state exists, we can interpret this passive information protection
in terms of a unitarily evolving subspace. This correspondence is
detailed in the next section.

\subsection{Perfect walls and DFS}

In many-body systems formed by a number of identical subsystems, the existence of a perfect wall state is equivalent to
the existence of a decoherence-free subspace. The latter are defined as follows \cite{lidarDecoherenceFreeSubspacesQuantum1998}:

\begin{definition}
Consider a Hilbert space $\cal H$ that may be decomposed as $\cal H \simeq H_\mathfrak{q} \oplus H_\mathfrak{r}$.
If any state $\rho$ that is initialized with support only on $\cal H_\mathfrak{q}$ evolves unitarily,
then $\cal H_\mathfrak{q}$ is said to be a Decoherence-Free Subspace (DFS).
\end{definition}

\noindent  Proposition \ref{thm:perfectwall} is closely related to the following well-known characterization of DFS \cite{lidarQuantumErrorCorrection2013,ticozziQuantumMarkovianSubsystems2008}.

\begin{theorem}
\label{thm:dfs}
Consider a bipartite system $\cal H = H_\mathfrak{s} \otimes H_\mathfrak{e}$ with Hamiltonian $H = H_\mathfrak{s} + H_\mathfrak{e} + H_\mathfrak{se}$.
Let $H_\mathfrak{se} = \sum_i A_i \otimes B_i$. If the following two conditions are satisfied, the system has a DFS
$\mathcal{H}_\mathfrak{q} = {\rm span}(\{\ket{\psi_k}\}_k)$
such that $\cal H_\mathfrak{s} \simeq H_\mathfrak{q} \oplus H_\mathfrak{r}$ (with unitary dynamics generated by $H_\mathfrak{s}$),
\begin{equation}
A_i \ket{\psi_k} = \lambda_i \ket{\psi_k},\ \forall i, \ket{\psi_k},
\end{equation}
\begin{equation}
{\rm supp}(H_\mathfrak{s} \ket{\psi}) \subseteq \mathcal{H}_\mathfrak{q},\ \forall \ket{\psi} \in \mathcal{H}_\mathfrak{q},
\end{equation}
i.e. if the operators $A_i$ have a common (degenerate) eigenspace that is left invariant by $H_\mathfrak{s}$.
\end{theorem}

In the example from the previous section, we may identify $A_1 = J^z_j$ and
$B_1 = J^z_{j+1}$ and verify that Theorem \ref{thm:dfs} confirms the existence of a DFS:
It suffices to check that the eigenspaces of $A_1$ are left invariant by $H_\mathfrak{s}$.
We formalize the general connection between DFS and perfect wall states in the following result.

\begin{proposition}
\label{thm:perfectwalldfs}
    If ${\rm dim}({\cal H}_\mathfrak{s}) = {\rm dim}({\cal H}_\mathfrak{l}) n_\mathfrak{w},$ for some integer $n_\mathfrak{w}$, then a perfect wall state exists if and only if a DFS exists.
\end{proposition}
\begin{proof}
Given a system with Hilbert space $\cal H = H_\mathfrak{s} \otimes H_\mathfrak{e}$
with a DFS $\cal H_\mathfrak{q}$ such that $\cal H \simeq (H_\mathfrak{q} \oplus H_\mathfrak{r}) \otimes H_\mathfrak{e}$,
and ${\rm dim}(\mathcal{H}_\mathfrak{r}) = n_\mathfrak{r} = n_\mathfrak{q} (n_\mathfrak{w} - 1)$,
where $n_\mathfrak{q} = {\rm dim}(\mathcal{H}_\mathfrak{q})$ and $n_\mathfrak{w} \in \mathbb{N}$.
In such a case, $\cal H_\mathfrak{r}$ can be interpreted as $n_\mathfrak{w}$ copies of $\cal H_\mathfrak{q}$,
i.e. ${\cal H}_\mathfrak{r} \simeq \bigoplus_{i=1}^{n_\mathfrak{w}-1} {\cal H}_\mathfrak{q}$.

In such a case, by choosing a basis $\{\ket{\phi_i}\} \otimes \{\ket{\psi_j}\}$ of $\mathcal{H}_\mathfrak{s}$,
such that ${\rm span}(\{\ket{\phi_i} \otimes \ket{\psi_0}\}_{i}) = \mathcal{H}_\mathfrak{q}$, we can identify
${\cal H}_\mathfrak{l} = {\rm span}\{ \ket{\phi_i}, i = 0, ..., n_\mathfrak{l}-1 \}$ and
${\cal H}_\mathfrak{w} = {\rm span}\{ \ket{\psi_i}, i = 0, ..., n_\mathfrak{w}-1 \}$.

Under this setting, we see that
\begin{equation*}
\cal H \simeq H_\mathfrak{l} \otimes H_\mathfrak{w} \otimes H_\mathfrak{e} \simeq H_\mathfrak{q} \oplus H_\mathfrak{r} \otimes H_\mathfrak{e},
\end{equation*}
and that $\ket{\psi_0}$ is a perfect wall state. More precisely, $\mathcal{H}_\mathfrak{l} \otimes {\rm span}(\ket{\psi_0}) \simeq \mathcal{H}_\mathfrak{q}$.
This means that initializing the state $\rho_\mathfrak{w}$ of $\cal H_\mathfrak{w}$
as the perfect wall state $\ketbra{\psi_0}{\psi_0}$ results in a system state $\rho_\mathfrak{s}$
such that $\rho_\mathfrak{s} \in \cal H_\mathfrak{q}$.
Conversely, if we initialize the state of the system inside of $\cal H_\mathfrak{q}$ and
then we trace out the logical system and the environment, we have $\tr_\mathfrak{le}(\rho) = \ketbra{\psi_0}{\psi_0}$.
\end{proof}

\section{General case: Finding an optimal wall subsystem}
\label{sec:subsystem}

Requiring that a perfect wall state (or a DFS) exists is too restrictive in realistic scenarios.
Therefore, in the following sections we propose a method to make these choices in general cases,
even when no perfect wall state exists, by means of Riemannian optimization.
This approach has been successfully used in a number of applications where optimization over fixed-dimension subspaces or unitary transformations is necessary, as done in e.g.\cite{casanovaFindingQuantumCodes2025, mansurogluVariationalHamiltonianSimulation2023, luchnikovRiemannianGeometryAutomatic2021}.

In order to treat general (approximate) cases, we first
search for a unitary $U = U_\mathfrak{s} \otimes \one_\mathfrak{e}$  that
allows us to identify the (virtual) degrees of freedom that mediate the dominant
interaction with the environment. After an optimal $\hat{U}$ is found, the controllable
subsystem can be decomposed into a logical and a wall subsystems,
${\cal H}_\mathfrak{s} \simeq {\cal H}_\mathfrak{l} \otimes {\cal H}_\mathfrak{w}$.
With this decomposition, we then look for a state of ${\cal H}_\mathfrak{w}$ that optimally
minimizes the interaction between ${\cal H}_\mathfrak{l}$ and ${\cal H}_\mathfrak{e}$.
In this section we propose an optimization problem to choose the optimal $\hat{U}$.
In light of the previous section, this can be seen as (the first step of) a method to {\em find approximate DFSs.} The next step shall be the selection of the wall step, which we will address later.

\subsection{Optimal wall subsystem via Riemannian optimization}

Let us denote the dimension of each subsystem ${\cal H}_\mathfrak{u}$ as $n_\mathfrak{u}$, where
$\mathfrak{u} \in \{\mathfrak{l}, \mathfrak{w}, \mathfrak{s}, \mathfrak{e}\}$.
In the same way, let $\one_\mathfrak{u}$ be the identity matrix of dimension $n_\mathfrak{u} \times n_\mathfrak{u}$.
Similarly, let $\{\sigma_i^\mathfrak{u}\}|_{i = 0}^{n_\mathfrak{u}^2 - 1}$ be
a basis of the space of Hermitian matrices of dimension $n_\mathfrak{u} \times n_\mathfrak{u}$,
such that $\sigma_0^\mathfrak{u} = \one_\mathfrak{u} / \sqrt{n_\mathfrak{u}}$ and $\tr(\sigma_i^\mathfrak{u} \sigma_j^\mathfrak{u}) = \delta_{ij}$.
Examples of such bases would be the Pauli or Gellmann matrices, including the identity
and normalized to one (instead of the usual normalization to two).
We then abbreviate the notation of tensor products of two or more basis elements as
$\sigma_{ij}^\mathfrak{uv} = \sigma_i^\mathfrak{u} \otimes \sigma_j^\mathfrak{v}$.
Finally, we implicitly specify the dimensions of an operator by choosing the
indices. For example, $H_\mathfrak{we}$ would stand for an operator of dimension $n_\mathfrak{w} n_\mathfrak{e} \times n_\mathfrak{w} n_\mathfrak{e}$.

The space of observables associated to Hilbert space ${\cal H}_\mathfrak{u}$ is denoted as
${\cal O}({\cal H}_\mathfrak{u})$, whereas the convex set of density matrices is denoted
as ${\mathfrak D}({\cal H}_\mathfrak{u})$.  Since the set of pure density matrices is the
set of extremal points of ${\mathfrak D}({\cal H}_\mathfrak{u})$ \cite{kimuraBlochVectorSpaceNLevel2005},
we use ${\rm extr}( {\mathfrak D}({\cal H}_\mathfrak{u}))$ to
denote it.

In order to identify the optimal choice of ${\cal H}_\mathfrak{w}$, we try to find
a unitary operator $U = U_\mathfrak{s} \otimes \one_\mathfrak{e}$, such that we can decompose
the system's Hamiltonian as
\begin{equation}
    U^\dagger H U = H_\mathfrak{l} + H_\mathfrak{w} + H_\mathfrak{e} + H_\mathfrak{lw} + H_\mathfrak{we} + \Delta
\end{equation}
where $\Delta = H_\mathfrak{le} + H_\mathfrak{lwe}$.
The choice of $U_\mathfrak{s}$ should be done such that it minimizes the norm of $\Delta$.
Let
\begin{equation}
    h_{ijk} = \tr(\sigma_{ijk}^\mathfrak{lwe} H),
\end{equation}
so that
\begin{equation}
    H = \sum_{ijk} h_{ijk} \sigma_{ijk}^\mathfrak{lwe}.
\end{equation}
Notice that we can then re-write $\Delta$ selecting only the relevant basis elements as follows:
\begin{equation}
    \Delta = \sum_{\substack{i,k > 0 \\ j \geq 0}}
    \tr(\sigma_{ijk}^\mathfrak{lwe} U^\dagger H U) \sigma_{ijk}^\mathfrak{lwe}.
\end{equation}
Using the orthonormality of the basis, we can now write
$\|\Delta\|^2 = \tr(\Delta^\dagger \Delta) = \tr(\Delta^2)=\sum_{\substack{i,k > 0 \\ j \geq 0}} | \tr(A_{ij} C_k) |^2$, where $A_{ij} = U_\mathfrak{s} \sigma_{ij}^\mathfrak{lw} U_\mathfrak{s}^\dagger$ and $C_k = \sum_{ab \geq 0} h_{abk} \sigma_{ab}^\mathfrak{lw}$. We can
then use this expression to define the following cost function
\begin{equation}
\label{eq:costfunc_sub}
    J(U_\mathfrak{s}) = \|\Delta\|^2
    = \sum_{\substack{i,k > 0 \\ j \geq 0}} | \tr(A_{ij} C_k) |^2
\end{equation}

We then propose the following optimization problem
on the special unitary group ${\rm SU}(n_\mathfrak{s})$,
\begin{equation}
\label{eq:optprob_sub}
    \hat{U}_\mathfrak{s} =
    \underset{U_\mathfrak{s} \in {\rm SU}(n_\mathfrak{s})}{\mathrm{argmin}} J(U_\mathfrak{s}).
\end{equation}
Notice that, at this point, no minimization is done on the terms $H_\mathfrak{lw}$ (nor $H_\mathfrak{we}$).
In the following subsection we introduce a regularization term to
limit any detrimental effect due to $H_\mathfrak{lw}$ that may occur due to the transformation $U_\mathfrak{s}$.

When equipped with the metric $\langle A, B \rangle = \tr(A^\dagger B)$,
the Lie group ${\rm SU}(n_\mathfrak{s})$ becomes a Riemannian manifold. 
This means that Problem \eqref{eq:optprob_sub} can be solved by
Riemannian gradient descent \cite{satoRiemannianOptimizationIts2021}.
The Riemannian gradient of \eqref{eq:costfunc_sub} can be computed following the
steps outlined in
\cite{wiersemaOptimizingQuantumCircuits2023, schulte-herbruggenGRADIENTFLOWSOPTIMIZATION2010},
resulting in
\begin{equation}
    {\rm grad} J(U_\mathfrak{s}) = -2 \sum_{\substack{i,k > 0 \\ j \geq 0}} \tr(A_{ij} C_k) [A_{ij}, C_k] U_\mathfrak{s},
\end{equation}
where $U_\mathfrak{s} \in {\rm SU}(n_\mathfrak{s})$ and ${\rm grad} J(U_\mathfrak{s}) \in T_{U_\mathfrak{s}} {\rm SU}(n_\mathfrak{s})$.
Therefore, the exponential map solving the gradient flow
$\dot{U}_\mathfrak{s}(t) = -{\rm grad} J(U_\mathfrak{s})$ is the following,
\begin{equation}
\label{eq:expmap_SUn}
    U_\mathfrak{s}(t + \epsilon) = {\rm exp}(-\Omega \epsilon) U_\mathfrak{s}(t)
\end{equation}
where $\Omega = -2 \sum_{\substack{i,k > 0 \\ j \geq 0}} \tr(A_{ij} C_k) [A_{ij}, C_k]$. More details on the implementation
of a numerical algorithm that solves this problem is given in \ref{sec:riemannianopt}.

After settling on a solution $\hat U_\mathfrak{s}$ and with $\hat U = \hat U_\mathfrak{s} \otimes \one_\mathfrak{e}$,
we use $\hat U$ to perform a change of basis. We call the new frame the ``$\hat{U}$-frame'',
and the Hamiltonian becomes
\begin{equation}
\begin{split}
\hat H =& \hat U^\dagger H \hat U = \sum_{ijk} g_{ijk} \sigma_{ijk}^\mathfrak{lwe} \\
=& \hat H_\mathfrak{l} + \hat H_\mathfrak{w} + \hat H_\mathfrak{e} + \hat H_\mathfrak{lw} + \hat H_\mathfrak{we} + \hat \Delta,
\end{split}
\end{equation}
where $g_{ijk} = \tr(\sigma_{ijk}^\mathfrak{lwe} \hat H)$. Equivalently,
by defining $\hat A_{ij} = \hat U_\mathfrak{s} \sigma_{ij}^\mathfrak{lw} \hat U_\mathfrak{s}^\dagger,$
we have that $g_{ijk}= \tr((\hat A_{ij} \otimes \sigma_k^\mathfrak{e}) H)$.

\subsection{Regularization}

The previous optimization problem allows us to find subsystem decompositions that
minimize the {\em direct} interaction between the logical subsystem and the environment.
However, these solutions are not unique, and some of them might lead to frames in which
the information loses coherence faster (in the absence of control) than in the original frame,
by increasing the interaction between $\cal{H}_\mathfrak{l}$ and $\cal{H}_\mathfrak{w}$.
In order to promote solutions in which the latter is kept low (or even provide
a full decoupling when possible), a regularization term can be added to the cost function.
In particular, we add the squared norm of $H_\mathfrak{lw}$, resulting in the following regularized
cost function,
\begin{equation}
J_{\rm reg}(U_\mathfrak{s}) = \|\Delta\|^2 + \eta_{\rm reg} \|H_\mathfrak{lw}\|^2,
\end{equation}
where $\eta_{\rm reg}$ is the regularization hyperparameter. Its Riemannian gradient is
\begin{equation}
\begin{split}
{\rm grad}J_{\rm reg}(U_s) =& -2 \sum_{\substack{i,k > 0 \\ j \geq 0}} \tr(A_{ij} C_k) [A_{ij}, C_k] U_\mathfrak{s}
-2 \sum_{i,j > 0} \tr(A_{ij} C_0) [A_{ij}, C_0] U_\mathfrak{s}.
\end{split}
\end{equation}
Then, the gradient flow is solved by the same exponential as before, substituting
$\Omega$ by $
\Omega_{\rm reg} = -2(
\sum_{\substack{i,k > 0 \\ j \geq 0}} \tr(A_{ij} C_k) [A_{ij}, C_k] +
\sum_{i,j > 0} \tr(A_{ij} C_0) [A_{ij}, C_0])
$.
The introduction of the regularization is further motivated analytically and tested in the example
given in Section \ref{sec:3spin_reg}.

\section{General case: Purity dynamics and the choice of wall state}
\label{sec:state}

After finding a good choice of wall subsystem, we proceed to search for an
optimal choice of initial wall state. The purity variation is used as quality index.
That is, given an initial product state $\rho = \rho_\mathfrak{l} \otimes \ketbra{w}{w} \otimes \rho_\mathfrak{e}$,
we search for a state $\ket{w}$ that minimizes the initial decay of the purity of the logical subsystem $\cal H_\mathfrak{l}$.

\subsection{Purity loss acceleration}

We wish to study the dynamics of the purity
$\gamma_\mathfrak{l}(\rho(t)) = \tr(\tr_\mathfrak{we}(\rho(t))^2) = \tr(\rho_\mathfrak{l}(t)^2)$ of the state
$\rho_\mathfrak{l}(t) = \tr_\mathfrak{we}(\rho(t)) \in {\mathfrak D}({\cal H}_\mathfrak{l})$,
where $\rho(t) \in {\mathfrak D}({\cal H})$. In particular, we are interested in the
initial dynamics of the purity when the initial state $\rho(0) = \rho_0$
is a product state. Our goal will then be to find a way to manipulate these dynamics
in order to extend the time that the purity remains sufficiently close to one,
as a way to reduce information loss to the environment.
The first derivative of the purity is as follows,
\begin{equation}
\dot\gamma_\mathfrak{l}(\rho) = 2 \tr(\tr_\mathfrak{we}({\cal L}(\rho))
\tr_\mathfrak{we}(\rho)).
\end{equation}
However, notice that if $\rho_0 = \rho_\mathfrak{l} \otimes \rho_\mathfrak{w} \otimes \rho_\mathfrak{e}$,
\begin{equation}
\tr_\mathfrak{we}([\sigma_{ijk}^\mathfrak{lwe}, \rho_0]) =
\tr(\sigma_j^\mathfrak{w} \rho_\mathfrak{w}) \tr(\sigma_k^\mathfrak{e} \rho_\mathfrak{e}) [\sigma_i^\mathfrak{l}, \rho_\mathfrak{l}],
\end{equation}
and if $L_m = \one_\mathfrak{l} \otimes \one_\mathfrak{w} \otimes L_m^\mathfrak{e}$,
we have $
\tr_\mathfrak{we}({\cal D}_{L_m}(\rho)) = 0.
$
This means that whenever the initial state is a product state, the initial
purity loss rate is equal to zero,
\begin{equation}
\begin{split}
\dot\gamma_\mathfrak{l}(\rho)|_{t=0} =& -2\iu \sum_{ijk} g_{ijk}
\tr(\sigma_j^\mathfrak{w} \rho_\mathfrak{w}) \tr(\sigma_k^\mathfrak{e} \rho_\mathfrak{e})
\tr([\sigma_i^\mathfrak{l}, \rho_\mathfrak{l}] \rho_l) = 0.
\end{split}
\end{equation}
Therefore, this quantity is not suitable for our aims, and we resort to the second derivative, which has
the following general expression:
\begin{equation}
\label{eq:2ndderiv}
\begin{split}
\ddot\gamma_\mathfrak{l}(\rho) =& 2 \tr(
\tr_\mathfrak{we}({\cal L} \circ {\cal L}(\rho)) \tr_\mathfrak{we}(\rho) +
\tr_\mathfrak{we}({\cal L}(\rho))^2
).
\end{split}
\end{equation}

Then, in order to find a function that depends only on the initial wall state $\rho_\mathfrak{w}$, we take the average of the initial
value of the second derivative $\ddot\gamma_\mathfrak{l}$ over all the initial pure logical states $\rho_\mathfrak{l}$ and fix
the initial state of the environment $\rho_\mathfrak{e}$.
This results in a function of $\rho_\mathfrak{w}$ only, whose expression is given in Proposition \ref{thm:avg_2nd_deriv}.

\begin{proposition}
\label{thm:avg_2nd_deriv}
Given a tripartite system with Hilbert space $\cal H = H_\mathfrak{l} \otimes H_\mathfrak{w} \otimes H_\mathfrak{e}$
and factorized initial state $\rho_0 = \rho_\mathfrak{l} \otimes \rho_\mathfrak{w} \otimes \rho_\mathfrak{e}$,
which follows Lindbladian dynamics with arbitrary Hamiltonian and noise operators $L_m = \one_\mathfrak{l} \otimes \one_\mathfrak{w} \otimes L^\mathfrak{e}_m$,
the second derivative of the purity $\gamma_\mathfrak{l}(\rho)$ has the following average initial value over all the reduced states $\rho_\mathfrak{l}$,
\begin{equation}
\label{eq:avg_2nd_deriv}
\langle \ddot\gamma_\mathfrak{l}(\rho)|_{t=0} \rangle_\mathfrak{l} =
4 \nu_n \sum_{\substack{i > 0 \\ bcjk \geq 0}} g_{ibc} g_{ijk}
(\tau_{c0} \tau_{0k} \tr(\sigma_b^\mathfrak{w} \rho_\mathfrak{w}) \tr(\sigma_j^\mathfrak{w} \rho_\mathfrak{w})
- \tau_{ck} \tr(\sigma_b^\mathfrak{w} \sigma_j^\mathfrak{w} \rho_\mathfrak{w})),
\end{equation}
where $\nu_n = \frac{1}{n_\mathfrak{l}} - \frac{1}{n_\mathfrak{l} (n_\mathfrak{l} + 1)}$ and
$\tau_{ck} = \tr(\sigma^\mathfrak{e}_c \sigma^\mathfrak{e}_k \rho_\mathfrak{e})$.
\end{proposition}
\begin{proof}
The proof is given in \ref{sec:avg_2nd_deriv}.
\end{proof}

\subsection{Wall states minimizing the purity loss acceleration}
A direct method to find a good choice of wall
state consists in solving an optimization problem that minimizes the
average initial purity loss acceleration \eqref{eq:avg_2nd_deriv}.
In particular, we can parameterize ${\rm extr} ({\cal O}({\cal H}_\mathfrak{w}))$
as the complex-valued $(n_\mathfrak{w}-1)$-sphere,
$S^{n_\mathfrak{w}-1}(\mathbb{C})$ and propose an optimization problem on it.
\begin{equation}
\label{eq:optprob_state_full}
\ket{\hat w} =
\underset{\ket{w} \in {\rm S}^{n_\mathfrak{w}-1}(\mathbb{C})}{\mathrm{argmin}}
\Gamma_1(\ket{w}),
\end{equation}
where $
\Gamma_1(\ket{w}) =
- (4 \nu_n)^{-1}
\langle \ddot \gamma_\mathfrak{l}(\rho)|_{t=0} \rangle_\mathfrak{l}
$.

Given a point $\ket{w} \in S^{n_\mathfrak{w}-1}(\mathbb{C})$ and vectors
$x,y$ in the tangent space $T_{\ket{w}} S^{n_\mathfrak{w}-1}(\mathbb{C})$,
the complex sphere $S^{n_\mathfrak{w}-1}(\mathbb{C})$ equipped with the metric induced
by the inner product
$\langle x, y\rangle_{\ket{w}} = x^\dagger (\one - \frac{1}{2} \ketbra{w}{w}) y$ is a Riemannian manifold.
Therefore, we can solve problem \eqref{eq:optprob_state_full}
with the same algorithm used to solve
\eqref{eq:optprob_sub}, but with the Riemannian gradient
\begin{equation}
\begin{split}
{\rm grad} \Gamma_1(\ket{w}) =&
2 \sum_{bcijk} g_{ibc} g_{ijk} (\tau_{c0} \tau_{0k}
(2 \langle \sigma_b^\mathfrak{w} \rangle_w \langle \sigma_j^\mathfrak{w} \rangle_w \ket{w} -
\langle \sigma_b^\mathfrak{w} \rangle_w \sigma_j^\mathfrak{w} \ket{w} \\
& \qquad\qquad\qquad\quad\ \ - \langle \sigma_j^\mathfrak{w} \rangle_w \sigma_b^\mathfrak{w} \ket{w}) -
\tau_{ck} (\langle \sigma_b^\mathfrak{w} \sigma_b^\mathfrak{w} \rangle_w \ket{w} - \sigma_b^\mathfrak{w} \sigma_j^\mathfrak{w} \ket{w})
),
\end{split}
\end{equation}
and exponential map
\begin{equation}
\label{eq:expmap_sph}
\ket{w(t + \epsilon)} = {\rm exp}((Q X \bra{w(t)} - \ket{w(t)} X^\dagger Q) \epsilon)
\ket{w(t)},
\end{equation}
where $Q = \one - \ketbra{w(t)}{w(t)}/2$ and $X = - {\rm grad} \Gamma_1(\ket{w(t)})$. More information on the implementation of
a numerical algorithm to solve this optimization problem is given in \ref{sec:riemannianopt}.

\subsection{Optimization with no direct LE interactions}
The previous optimization problem becomes significantly simpler in the absence of direct Hamiltonian coupling between the logical subsystem and the environment.
In fact, assuming that $\hat \Delta = 0$, the second derivative becomes
\begin{equation}
\begin{split}
\ddot\gamma_\mathfrak{l}(\rho)|_{\substack{t=0 \\ \hat \Delta = 0}} =&
\frac{-4}{n_\mathfrak{e}} \sum_{abij > 0} g_{ab0} g_{ij0}
\tr(\sigma_a^\mathfrak{l} \sigma_i^\mathfrak{l} \rho_\mathfrak{l}^2 - \sigma_a^\mathfrak{l} \rho_\mathfrak{l} \sigma_i^\mathfrak{l} \rho_\mathfrak{l})
{\rm Cov}_{\rho_\mathfrak{w}}(\sigma_b^\mathfrak{w}, \sigma_j^\mathfrak{w}),
\end{split}
\end{equation}
where $
{\rm Cov}_{\rho_\mathfrak{w}}(\sigma_b^\mathfrak{w}, \sigma_j^\mathfrak{w}) =
\tr(\sigma_b^\mathfrak{w} \sigma_j^\mathfrak{w} \rho_\mathfrak{w}) -
\tr(\sigma_b^\mathfrak{w} \rho_\mathfrak{w}) \tr(\sigma_j^\mathfrak{w} \rho_\mathfrak{w})
$.

Notice that, as before, the terms with $a = 0$ and/or $i = 0$ do not contribute to the sum. However, now this is also
the case if $b = 0,$ and / or $j = 0$. Therefore, the expression is simplified by restricting all four
indices to be greater than zero.
In other words, this expression is completely independent of all local Hamiltonian
terms and $\hat H_\mathfrak{we}$ and depends only on the interaction term $\hat H_\mathfrak{lw}$.

Since the simplified initial purity loss acceleration depends only on a
single bipartite operator, $\hat H_\mathfrak{lw}$, we can take its
operator-Schmidt decomposition (OSD). This allows us to propose a choice
of wall state with an analytic expression by targeting the component of
$\hat H_\mathfrak{lw}$ with the highest singular value. In fact, the LW interaction Hamiltonian can be written as
\begin{equation}
\hat H_\mathfrak{lw} = \sum_{ij > 0} g_{ij0} \sigma_{ij}^\mathfrak{lw}.
\end{equation}
We can then arrange the coefficients into a matrix
$G = [0] \oplus [g_{ij0}]_{ij>0} = [\tr(\sigma_{ij}^\mathfrak{lw} \hat H_\mathfrak{lw})]_{ij \geq 0}$ and
find its SVD $G = U S V^\dagger$, where
$U \in \mathbb{R}^{n_\mathfrak{l}^2 \times n_\mathfrak{u}^2}$,
$V \in \mathbb{R}^{n_\mathfrak{w}^2 \times n_\mathfrak{u}^2}$,
$S \in \mathbb{R}^{n_\mathfrak{l}^2 \times n_\mathfrak{u}^2}$ and
$n_\mathfrak{u} = \min \{n_\mathfrak{l}, n_\mathfrak{w}\}$.
Due to the properties of the SVD, $U$ and $V$ are such that $U^\dagger U = V^\dagger V = \one_\mathfrak{u}$,
and $S$ is a diagonal matrix with semipositive entries ordered in non-increasing order, i.e.
$S = {\rm diag}(s_1, \cdots, s_{n_\mathfrak{u}^2})$ where $s_1 \geq \cdots \geq s_{n_\mathfrak{u}^2}$.
Additionally, because $\hat H_\mathfrak{lw}$ is a purely interaction Hamiltonian (i.e. it does not contain any local term),
we have that for all $i > 1$ and $j < n_\mathfrak{u}^2$ the entries $U_{i n_\mathfrak{u}^2} = U_{1j} = V_{i n_\mathfrak{u}^2} = V_{1j} = S_{n_\mathfrak{u}^2 n_\mathfrak{u}^2} = 0$, while
$U_{1 n_\mathfrak{u}^2} = V_{1 n_\mathfrak{u}^2} = 1$.

By defining new operators $C_i = \sum_{a \geq 1} U_{ai} \sigma_{a-1}^\mathfrak{l}$ and $D_j = \sum_{b \geq 1} V_{bj} \sigma_{b-1}^\mathfrak{w}$
we can rewrite the CW interaction Hamiltonian as
\begin{equation}
\label{eq:osd_ham}
\hat H_\mathfrak{lw} = \sum_{i \geq 1} s_i C_i \otimes D_i,
\end{equation}
which is its OSD. Notice that due to the orthonormality of the columns of $U$ and $V$,
the two sets of matrices $\{C_i\}_i$ and $\{D_i\}_i$ satisfy $\tr(C_i C_j) = \delta_{ij}$
and $\tr(D_i D_j) = \delta_{ij}$, just as was the case for $\{\sigma_i^\mathfrak{l}\}_i$
and $\{\sigma_j^\mathfrak{w}\}_j$.
We also have $\tr(C_i) = 0$ and $\tr(D_i) = 0$ for $i < n_\mathfrak{u}^2$ and $C_{n_\mathfrak{u}^2} = \sigma_0^\mathfrak{l}$ and $D_{n_\mathfrak{u}^2} = \sigma_0^\mathfrak{w}$.

Writing $\ddot\gamma_\mathfrak{l}(\rho)|_{t=0}$ in terms of the OSD, we obtain
\begin{equation}
\begin{split}
\ddot\gamma_\mathfrak{l}(\rho)|_{\substack{t=0 \\ \hat \Delta = 0}} =& 
\frac{-4}{n_\mathfrak{e}} \sum_{ij} s_i s_j \tr(C_i C_j \rho_\mathfrak{l}^2 - C_i \rho_\mathfrak{l} C_j \rho_\mathfrak{l})
{\rm Cov}_{\rho_\mathfrak{w}}(D_i, D_j),
\end{split}
\end{equation}
and its average over the pure states of ${\cal H}_\mathfrak{l}$ and ${\cal H}_\mathfrak{e}$ is
\begin{equation}
\begin{split}
\langle \ddot\gamma_\mathfrak{l}(\rho)|_{\substack{t=0 \\ \hat \Delta = 0}} \rangle_\mathfrak{le} =&
- 4 \frac{\nu_n}{n_\mathfrak{e}}
\sum_i s_i^2 {\rm Var}_{\rho_\mathfrak{w}}(D_i).
\end{split}
\end{equation}

Under these assumptions (no LE or no three-body interaction terms) the optimization problem  \eqref{eq:optprob_state_full} simplifies to:
\begin{equation}
\label{eq:optprob_state}
\ket{\hat w} =
\underset{\ket{w} \in {\rm S}^{n_\mathfrak{w}-1}(\mathbb{C})}{\mathrm{argmin}} \Gamma_2(\ket{w}),
\end{equation}
where $
\Gamma_2(\ket{w}) = - (4 \frac{\nu_n}{n_\mathfrak{e}})^{-1} \langle \ddot\gamma_\mathfrak{l}(\rho)|_{\substack{t=0 \\ \hat \Delta = 0}} \rangle_\mathfrak{le}
$, i.e.
\begin{equation}
\label{eq:osd_gamma}
\Gamma_2(\ket{w}) = \sum_i s_i^2 {\rm Var}_{\ket{w}}(D_i).
\end{equation}
Problem \eqref{eq:optprob_state} can also be considered as a simpler, yet approximate alternative to the general case.

We can solve problem \eqref{eq:optprob_state}
with the same Riemannian gradient descent algorithm  used to solve
\eqref{eq:optprob_state_full}, but with the following Riemannian gradient
\begin{equation}
\label{eq:osd_gamma_grad}
\begin{split}
{\rm grad} \Gamma_2(\ket{w}) =& \sum_i s_i^2 (2 (D_i^2 \ket{w} - \bra{w} D_i^2 \ket{w} \ket{w}) - \\
&\qquad 4 (\bra{w} D_i \ket{w} D_i \ket{w} - (\bra{w} D_i \ket{w})^2 \ket{w})),
\end{split}
\end{equation}
and with $X = - {\rm grad} \Gamma_2(\ket{w(t)})$. See \ref{sec:riemannianopt} for a more detailed
explanation of Riemannian gradient descent method.

\subsection{Analytical solution for qubit walls}

If the wall subsystem is a qubit system, i.e. ${\cal H}_\mathfrak{w} \simeq \mathbb{C}^2$, we can
provide an analytical solution to the optimization problem \eqref{eq:optprob_state}.
In particular, the solution is given by the eigenstates of the dominant interaction term $D_1$,
i.e. wall states $\ket{\hat w}$ such that $D_1 \ket{\hat w} = \lambda \ket{\hat w}$.

\begin{proposition}
\label{thm:optimwall}
    Given $H_\mathfrak{lw}$ in the OSD form as in (\ref{eq:osd_ham}), $n_\mathfrak{w} = 2$
    and $s_1 > s_i$ for $i > 1$, choosing $\ket{w} \in \mathrm{eigenvectors}(D_1)$
    minimizes $\Gamma_2(\ket{w})$ over $\bigcup_{i = 1}^3 \mathrm{eigenvectors}(D_i)$,
    and locally minimizes $\Gamma_2(\ket{w})$ over the set of pure states of ${\cal H}_\mathfrak{w}$
\end{proposition}
\begin{proof}
The proof is given in \ref{sec:qubitproof}.
\end{proof}

If $\Delta = 0$ and the chosen $\ket{\hat{w}}$ is a perfect wall state,
initializing the state of $\cal H_\mathfrak{w}$ to $\ket{\hat w}$ would cause the logical subsystem
to evolve unitarily, i.e. it would provide an effective decoupling of $\cal H_\mathfrak{l}$
from the rest of the system. However, if that is not the case, one could try different kinds
of stabilization methods that could be used to further slow the loss of purity by acting only
on the wall subsystem and not on the logical subsystem. Three different control strategies are proposed
in the next section.

\section{General Case: Stabilizing and Protecting the Wall State}
\label{sec:enc_prot}

After settling on a choice of wall subsystem (as in Section \ref{sec:subsystem}) and wall state (as in Section \ref{sec:state}),
we want to improve the performance of the noise suppression protocol by stabilizing the wall state,
so that the optimization of the purity dynamics remains significant for longer times, with the purity $\gamma_\mathfrak{l}(\rho(t))$
remaining close to one for longer times.

The first method to stabilize the initial state that we shall consider is {\em repeated measurements}:
If the frequency of the measurements is high enough, they induce a quantum Zeno effect
\cite{misraZenoParadoxQuantum1977}, slowing the drift of the wall state.
Next, we consider the addition of {\em engineered dissipation} terms to the wall, so that the desired wall state is approximately stabilized. Lastly, we will consider constant {\em strong Hamiltonian driving} as a
way to energetically decouple the wall subsystem from the other parts of the system. This last
approach, which also results in zeno-like dynamics, exhibits interesting phenomena, in which the purity oscillations may not just
slow down, but also shrink in amplitude with an increase of the control Hamiltonian amplitude. This phenomenon may, in certain cases, lead to
lower bounds on purity $\gamma_\mathfrak{l}(\rho(t))$ that are valid for {\em at any time}. These might be seen as an approximate {\em eternal purity preservation}: see Figure \ref{fig:control_hamd}, and Section \ref{sec:eternal} for an analytical treatment.
All these control methods are implemented and investigated numerically in the following subsections for prototypical examples.

\subsection{Repeated measurements}

We assume to be able to measure
$\hat \Pi = \one_\mathfrak{l} \otimes \ketbra{\hat w}{\hat w} \otimes \one_\mathfrak{e}$.
As such, if the measurement result is positive, the state $\rho$ collapses to
$\hat\Pi \rho \hat\Pi$, otherwise we obtain $(\one - \hat\Pi) \rho (\one - \hat\Pi)$.
\begin{equation}
\rho(t_i + \varepsilon) = \hat\Pi \rho(t_i) \hat\Pi + \hat\Pi^\perp \rho(t_i) \hat\Pi^\perp,
\end{equation}
where $\hat\Pi^\perp = (\one - \hat\Pi)$.
Measurements are performed with frequency $f$. We can then discretize the evolution as
\begin{equation}
\rho(t_{i+1}) = \hat\Pi U_{1/f} \rho(t_i) U_{1/f}^{\dagger} \hat\Pi +
\hat\Pi^\perp U_{1/f} \rho(t_i) U_{1/f}^{\dagger} \hat\Pi^\perp,
\end{equation}
where $U_{1/f} = \exp(-\iu \hat H / f)$ and $t_{i} = i / f$. In the limit where
$f \to \infty$, the system behaves according to the following effective Zeno
Hamiltonian \cite{facchiThreeDifferentManifestations2003},
\begin{equation}
\label{eq:zenoham}
H_Z = \Pi H \Pi + \Pi^\perp H \Pi^\perp.
\end{equation}
The resulting evolution operator is
\begin{equation}
e^{-\iu H_Z t} 
= \Pi e^{-\iu \Pi H \Pi t} \Pi + \Pi^\perp e^{-\iu \Pi^\perp H \Pi^\perp t} \Pi^\perp,
\end{equation}
which clearly prevents any entanglement between the wall and the other partitions,
whenever the initial state has the form
$\rho_0 = \rho_\mathfrak{s} \otimes \ketbra{\hat w}{\hat w} \otimes \rho_\mathfrak{e}$,
as it leaves the wall partition invariant.

\subsection{Dissipation}

If one could design dissipation terms acting on the wall subsystem by techniques such as
reservoir engineering and dissipative control \cite{altafiniModelingControlQuantum2012,ticozziAnalysisSynthesisAttractive2009}, one could implement dissipation operators such as
\begin{equation}
L_i = \one_\mathfrak{l} \otimes \ketbra{\hat w}{\hat w^\perp_i} \otimes \one_\mathfrak{e},
\end{equation}
where $\{\ket{\hat w^\perp_i}\}_i^{n_\mathfrak{w} - 1}$ are $n_\mathfrak{w} - 1$ mutually orthonormal
states such that $\braket{\hat w}{\hat w^\perp_i} = 0, \forall i$. The state dynamics would then
be governed by
\begin{equation}
\dot\rho(t) = - \mathrm{i} [\hat H, \rho(t)] +
\eta \sum_i (L_i \rho L_i^\dagger - \frac{1}{2} \{L_i^\dagger L_i, \rho\}).
\end{equation}
One can prove \cite{ticozziQuantumMarkovianSubsystems2008,ticozziHamiltonianControlQuantum2012} that the desired wall state
$\ketbra{\hat{w}}{\hat{w}}$ is the unique fixed point of the reduced dynamics $\tr_\mathfrak{le}(\rho(t)),$ and hence it is asymtpotically stable. For increasing $\gamma,$ the convergence becomes faster.

\subsection{Strong Hamiltonian driving}

Let there be a controllable Hamiltonian term $\one_\mathfrak{l} \otimes H_u \otimes \one_\mathfrak{e}$
whose amplitude can be modulated via an input function $u(t)$. We are interested in studying the case where $u(t) = \kappa$ is large. The driven
Hamiltonian would be the following
\begin{equation}
H_\kappa(t) = \hat H + u(t) \one_\mathfrak{l} \otimes H_u \otimes \one_\mathfrak{e}
\end{equation}
\begin{equation}
H_\kappa = \hat H + \kappa \one_\mathfrak{l} \otimes H_u \otimes \one_\mathfrak{e}.
\end{equation}
Assume also that $H_u$ is nondegenerate. If $\kappa$ is large enough, this kind of control
would consist of a detuning of the wall from the rest of the system, i.e. the eigenenergies
of $H_\mathfrak{w}$ are taken away from any possible resonance with the eigenenergies of the
other Hamiltonian terms (except for any possible eigenenergy fixed at zero).

In particular, we want the control Hamiltonian to be of the form
\begin{equation}
H_u = \lambda_{\hat w} \ketbra{\hat w}{\hat w} + \sum_{j=1}^{n_\mathfrak{w} - 1} \lambda_{\hat w^\perp_j} \ketbra{\hat w^\perp_j}{\hat w^\perp_j},
\end{equation}
where $\ket{\hat{w}^\perp_j}$ are $n_\mathfrak{w} - 1$ mutually orthonormal states such that $\braket{\hat{w}^\perp_j}{\hat{w}} = 0$.
By the dynamical Zeno theorem \cite{facchiThreeDifferentManifestations2003}, we can then prove that in the limit $\kappa \to \infty$,
the dynamics of the system are governed by the following effective Hamiltonian
\begin{equation}
H_{\rm eff} = \ketbra{\hat w}{\hat w} \hat{H} \ketbra{\hat w}{\hat w} +
\sum_{j=1}^{n_\mathfrak{w} - 1} \ketbra{\hat w^\perp_j}{\hat w^\perp_j} \hat{H} \ketbra{\hat w^\perp_j}{\hat w^\perp_j}.
\end{equation}
Under this Hamiltonian, the wall state $\ket{\hat w}$ becomes a fixed point and therefore no entanglement can be formed
with the rest of the system.

\section{Numerical Results}\label{sec:examples}

In this section, we provide a series of examples to demonstrate the effectiveness of our method. We start by considering two minimal three-qubit systems. This choice allows us to test the algorithms against a model with a simple known solution and highlight the importance of the regularization terms.
First, we introduce a toy Hamiltonian in order to test the algorithm to find the optimal subsystem
decomposition of $\mathcal{H}_\mathfrak{s}$ independently of the choice of wall state. Next ,
we consider a three qubit Ising chain to test our method to choose the optimal wall state
and the stabilizing controls, independently of the choice of subsystem decomposition.

The whole method is then tested in two additional, richer examples. The first is a {\em spin lattice
with three-body interaction terms.} The other is a {\em central spin model,} in which a single 3/2-spin is decomposed
into logical and wall subsystems.

In all of the examples, we use a regularization hyperparameter of $\eta_{\rm reg} = 0.01$ and initialize
the environment to a thermal state with inverse temperature $\beta = 0.01$,
\begin{equation}
\rho_\mathfrak{e} = \frac{e^{-H_\mathfrak{e} \beta}}{\tr(e^{-H_\mathfrak{e} \beta})}.
\end{equation}
The starting point of the Riemannian gradient descent for subsystem decomposition is set to a random unitary,
whereas that for wall state optimization is set to $\ket{w_0} = \mu (\ket{w_D} + \ket{w_r}/10)$, where $\ket{w_D}$
is an eigenstate of $D_1$, $\ket{w_r}$ is a random unit vector and $\mu$ is a normalization constant.

\subsection{Three-Spin Toy Model: The Effect of Regularization} \label{sec:3spin_reg}
\subsubsection{Model and analytical discussion:}  Consider a three-qubit system and the following toy Hamiltonian:
\begin{equation}
\label{eq:toyham_reg}
H_0 = \sum_{i=1}^3 J_i^z + J_1^x J_3^x,
\end{equation}
where $J_i^z$ is given by \eqref{eq:Jz}, whereas $J_i^x$ is given by
\begin{equation}
\label{eq:Jx}
J_i^x = \one^{\otimes i-1} \otimes \frac{1}{2}\sigma_x \otimes \one^{\otimes N-i},
\end{equation}
and $N = 3$. Additionally, consider the following two unitary operators:
\begin{equation}
\label{eq:U1}
U_{1,\mathfrak{s}} = \ketbra{00}{11} + \ketbra{01}{00} + \ketbra{10}{10} + \ketbra{11}{01},
\end{equation}
\begin{equation}
\label{eq:U2}
U_{2,\mathfrak{s}} = \ketbra{00}{00} + \ketbra{01}{10} + \ketbra{10}{01} + \ketbra{11}{11}.
\end{equation}

\noindent Then, applying them to the Hamiltonian \eqref{eq:toyham_reg}, we obtain
\begin{equation}
H_1 = 2 J_1^z J_2^z - J_1^z + J_3^z + J_2^x J_3^x,
\end{equation}
\begin{equation}
H_2 = J_2^z + J_1^z + J_3^z + J_2^x J_3^x.
\end{equation}
where $H_i = (U_{i, \mathfrak{s}} \otimes \one)^\dagger H_0 (U_{i, \mathfrak{s}} \otimes \one)$.
Both $H_1$ and $H_2$ satisfy $\Delta = 0$ and therefore
$J(U_{1, \mathfrak{s}}) = J(U_{2, \mathfrak{s}}) = 0$. However, in the frame associated with the Hamiltonian
$H_1$ the logical subsystem remains coupled to the environment via the wall,
whereas in the frame of $H_2$ it is perfectly decoupled. This difference is
captured by the regularized cost function. Indeed,
$J_{\rm reg}(U_{1, \mathfrak{s}}) > 0$, whereas $J_{\rm reg}(U_{2, \mathfrak{s}}) = 0$.

\subsubsection{Numerical construction of the wall subsystem:} The first test for our method is to try to decompose the first two subsystems into a logical and a wall {\em virtual} subsystems using our numerical approach.
In order to test the validity of our algorithm and the effectiveness of the regularization,  we ran the optimization procedure twice, first
without and then with regularization. The results are shown in Table \ref{tab:subsystem_choice}.

\begin{table}[!htpb]
\centering
\begin{tabular}{c|c|c}
\hline
    & $J(U_\mathfrak{s})$ & $J_{\rm reg}(U_\mathfrak{s})$ \\
\hline
$\hat{U}_{\rm reg}$ & $3.87 \times 10^{-15}$ & $7.76 \times 10^{-8}$ \\
$\hat{U}$           & $3.53 \times 10^{-19}$ & $0.0367$ \\
$\one$              & $0.5$                  & $0.5$     \\
\hline
\end{tabular}
\caption{Value of the cost function for subsystem decomposition after optimizing with and without regularization and Hamiltonian \eqref{eq:toyham_reg}.
The regularization hyper-parameter was set to $0.01$. $\hat{U}$ refers to the numerical solution without regularization, whereas
$\hat{U}_{\rm reg}$ refers to the solution with regularization.}
\label{tab:subsystem_choice}
\end{table}

As expected, the algorithm can remove the direct coupling between the logical subsystem and the environment.
However, without regularization, it does so in a way that leaves a significant coupling between the virtual logical subsystem and the wall,
which is greater in magnitude than the original coupling with the environment. In fact, for $U = \one$, we had $\|H_\mathfrak{le}\|^2 = 0.5$,
and for $U = \hat{U}$ we instead have $\|H_\mathfrak{lw}\|^2 = 3.67$.
However, with regularization, it is able to find a decomposition of $\cal H_\mathfrak{s}$ such that $\cal H_\mathfrak{l}$ is completely
decoupled from the rest of the system.

\begin{figure}[ht]
\centering
\includegraphics[width=0.7\linewidth]{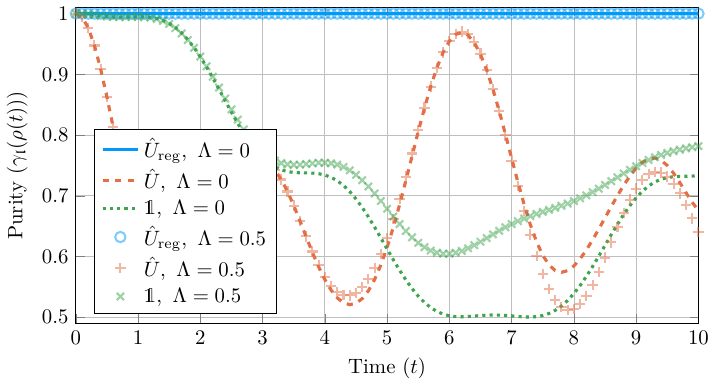}
\caption{Sample purity dynamics of the system with Hamiltonian \eqref{eq:toyham_reg}. The solid blue line shows the purity
$\gamma_\mathfrak{l}$ of the logical subsystem in the optimized frame with regularization. The dashed red line shows the purity of the logical
subsystem in the optimized frame without regularization. The dotted green line shows the purity of the logical subsystem in the original frame.
The markers of matching colors show the purity of the system when a Lindbladian pumping term \eqref{eq:pumping} is added to the environment.}
\label{fig:subsystem_choice}
\end{figure}

These findings are confirmed by the simulations shown in Figure \ref{fig:subsystem_choice}, where we plot sample purity dynamics with the same random initial logical and wall state in
three different subsystem decompositions given by the regularized solution, the nonregularized solution, and the original
frame. The initial state of the environment is set to a thermal state with inverse temperature $\beta = 0.01$.
Additional purity trajectories are plotted,  where a Lindbladian term was added to the environment, inducing spontaneous decay with a noise operator of the following form
\begin{equation}
\label{eq:pumping}
L = \Lambda J^+_3 = \Lambda (J^x_3 + \iu J^y_3),
\end{equation}
where $\Lambda = 0.5$, $J^x_i$ is given by \eqref{eq:Jx} and $J_i^y$ is given by
\begin{equation}
\label{eq:Jy}
J_i^y = \one^{\otimes i-1} \otimes \frac{1}{2}\sigma_y \otimes \one^{\otimes N-i}.
\end{equation}
We see that the inclusion of the regularization allows our optimization algorithm to effectively decouple
the logical subsystem $\mathcal{H}_\mathfrak{l}$ from the rest of the system, and therefore the purity
is perfectly preserved. However, without regularization, the purity decays even faster than
in the original frame. This is possible because, in spite of the removal of the coupling to the environment,
a stronger coupling of the system to the wall partition can be allowed, and at this stage we are not yet applying any initialization or
control to the wall partition. The introduction of Markovian dissipation in the third qubit has little to no effect on the performance of the method.

\subsection{Three-spin Ising model: Choice of wall state}

We shift our focus now to the choice of initial wall state.
To test our method, we consider the following Ising Hamiltonian with transversal field:
\begin{equation}
\label{eq:ising_trans}
H = \sum_{i=1}^3 J_i^z + \sum_{i=1}^2 J_i^x J_{i+1}^x,
\end{equation}
modeling a chain of three spins. Each spin has a Hilbert space $\mathcal{H}_i = \mathbb{C}^2$,
and the Hilbert space of the full chain is $\mathcal{H} = \bigotimes_{i=1}^3 \mathcal{H}_i$.
In this subsection, we wish to demonstrate the method to choose the initial wall state
independently of the previously shown method to optimize the subsystem decomposition of
$\mathcal{H}_\mathfrak{s}$. Therefore, we take $\mathcal{H}_\mathfrak{l} = \mathcal{H}_1$,
$\mathcal{H}_\mathfrak{w} = \mathcal{H}_2$ and $\mathcal{H}_\mathfrak{e} = \mathcal{H}_3$.

Notice that in this case we have that $\Delta = 0$, therefore the cost functions
$\Gamma_1$ and $\Gamma_2$ are identical. Moreover, the LW interaction Hamiltonian
is equal to $H_\mathfrak{lw} = J_1^x J_2^x$, i.e. we have a single nonzero singular
value stemming from the OSD decomposition, as evidenced in Table \ref{tab:state_singvals}.

\begin{table}[!htpb]
\centering
\begin{tabular}{c|c|c|c|c}
\hline
    & $s_1$ & $s_2$ & $s_3$ & $s_4$  \\
\hline
$\one$    & 0.5 & 0.0 & 0.0  & 0.0 \\
\hline
\end{tabular}
\caption{Singular values of the OSD decomposition of $H_\mathfrak{lw}$ for Hamiltonian \eqref{eq:ising_trans}.}
\label{tab:state_singvals}
\end{table}

Since the wall is a qubit system, the eigenvalues of $J_x$ are solutions to the
optimization problem \eqref{eq:optprob_state}. In fact, we get $\ket{\hat{w}} = \ket{-}$ as the output
of the gradient descent algorithm.

Figure \ref{fig:state_sample} shows the sample purity dynamics for a random initial pure state in the logical subsystem
and a thermal state in the environment. As expected, we observe a slower purity decay with the optimized
wall state. We note that also in this case the inclusion of the Lindbladian term \eqref{eq:pumping}
has no effect on the purity of the logical subsystem.

\begin{figure}[ht]
\centering
\subfloat[][Sample purity dynamics.]{
\includegraphics[width=0.45\linewidth]{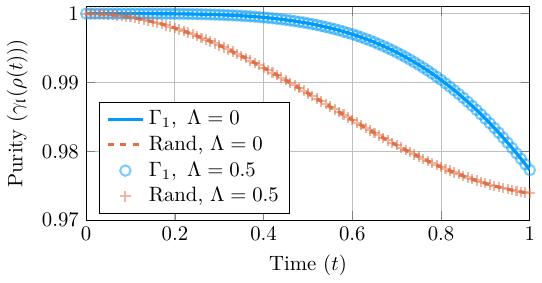}
\label{fig:state_sample}
}
\subfloat[][Average advantage of $\ket{\hat{w}}$.]{
\includegraphics[width=0.45\linewidth]{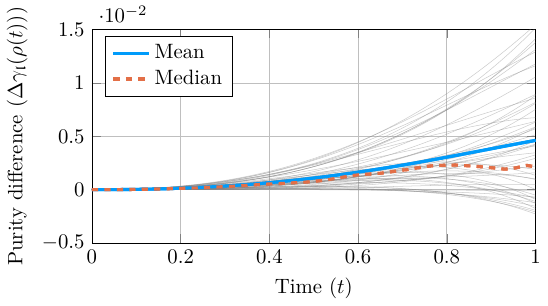}
\label{fig:state_avg}
}
\caption{\subref{fig:state_sample} Sample purity dynamics for different initial wall states. The solid blue line
shows the purity $\gamma_\mathfrak{l}$ for the optimal choice $\ket{\hat{w}} = \ket{-}$, whereas
the dashed red line shows the purity for a random choice of wall state.
The markers of matching colors show the purity of the system when a Lindbladian pumping term \eqref{eq:pumping} is added to the environment.
\subref{fig:state_avg} Difference between the purity of the logical subsystem between realizations
with the optimal wall state and a random wall state. 50 samples are shown in gray, whose
mean is represented by the solid blue line, whereas the median is represented by the dashed red line. Positive values indicate that the wall state is performing better than the random one.
}
\label{fig:state}
\end{figure}

At this point, we would like to highlight that the advantage of the optimal wall state (when
no perfect wall state exists) is on average. That is, for certain initial states $\rho_\mathfrak{l}$
of the logical subsystem, a different choice of wall state $\ket{w}$ may provide a slower purity decay.
But since we want to provide good protection for all initial, and possibly unknown, $\rho_\mathfrak{l}$,
it is the average that concerns us. Figure \ref{fig:state_avg} shows the difference in purity
$\gamma_\mathfrak{l}$ between realizations with the optimal wall state $\ket{-}$ and a random wall state.

\subsubsection{Stabilizing control schemes:}
We next test the effectiveness of the different control schemes with the same Hamiltonian
\eqref{eq:ising_trans} and wall state $\ket{\hat{w}} = \ket{-}$ derived in the previous subsection.

Figure \ref{fig:control_meas} shows the purity dynamics with repeated
measurements with projectors $\hat{\Pi} = \ketbra{-}{-}$ and $\hat{\Pi}^\perp = \ketbra{+}{+}$,
which project the wall state back to (mixtures of) the eigenstates of $J_x$.
Figure \ref{fig:control_diss} shows the purity dynamics with a single Lindblad operator
$L = \ketbra{-}{+}$, which causes the wall state to decay back to its initial state.
In both cases, we observe an effective slowdown of the purity loss.
In Figure \ref{fig:control_hamd} we see the purity dynamics with the strong Hamiltonian driving,
using $H_u = \ketbra{-}{-} - \ketbra{+}{+} = - 2 J_x$.
This case is interesting because, more than a slowdown of the purity dynamics, what we observe
is a shrinking of the oscillations. The entire dynamics is compressed towards constant purity equal to one.
We study this behavior in more depth in Section \ref{sec:eternal}.

Figure \ref{fig:control_times} shows the time that it takes the purity $\gamma_\mathfrak{l}$ to reach
a value of $0.97$. Here we confirm that, in the case of Hamiltonian driving,
after a certain value of $\kappa$, the purity remains bounded above $0.97$ for the entirety of the simulation time.
In the case of dissipative control, the time that it takes $\gamma_\mathfrak{l}$ to reach $0.97$ is faster than
linearly on $\eta$. For repeated measurements, we observe a slow growth that is approximately
piecewise linear on $f$, with some discontinuous jumps that occur when the valleys of oscillations of the purity decay
cross the $0.97$ threshold.

\begin{figure}[ht]
\centering
\subfloat[][Repeated measurements.]{
\includegraphics[width=0.45\linewidth]{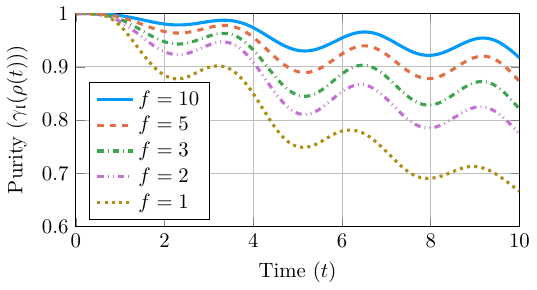}
\label{fig:control_meas}
}\
\subfloat[][Dissipation.]{
\includegraphics[width=0.45\linewidth]{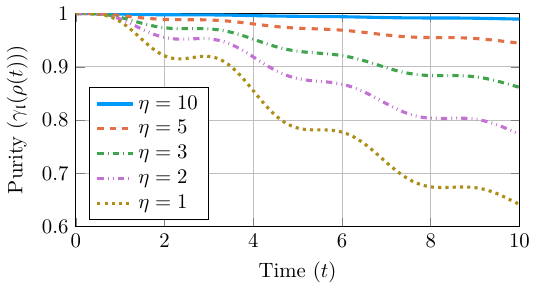}
\label{fig:control_diss}
}\
\subfloat[][Hamiltonian driving.]{
\includegraphics[width=0.45\linewidth]{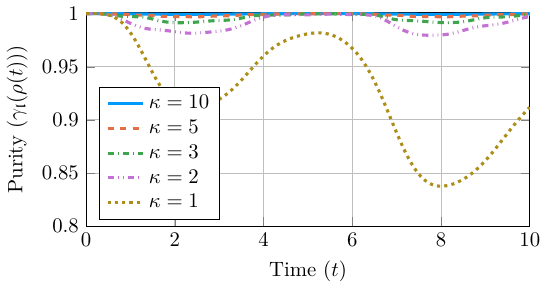}
\label{fig:control_hamd}
}\
\subfloat[][Time to reach $\gamma_\mathfrak{l} = 0.97$.]{
\includegraphics[width=0.45\linewidth]{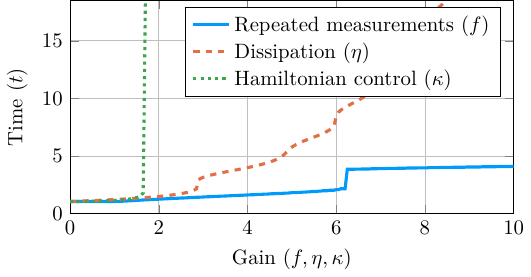}
\label{fig:control_times}
}
\caption{\subref{fig:control_meas}, \subref{fig:control_diss}, \subref{fig:control_hamd} Purity dynamics when each of the proposed controls is applied for different values of the corresponding control gain variable. The Hamiltonian of the system is given by \eqref{eq:ising_trans} and the wall state is $\ket{-}$. In \subref{fig:control_meas} the measurement projectors are $\Pi = \ketbra{-}{-}$ and $\Pi^\perp = \ketbra{+}{+}$. In \subref{fig:control_diss} the dissipation operator is $L = \ketbra{-}{+}$. In \subref{fig:control_hamd} the control Hamiltonian is $H_u = \ketbra{-}{-} - \ketbra{+}{+}$. \subref{fig:control_times} shows the time that it takes the purity $\gamma_\mathfrak{l}$ to reach a value of $0.97$ for different values of the corresponding control variable of the three control schemes. The green dotted line, corresponding to the strong Hamiltonian driving, goes off the plot for $\kappa > 2$ because it never reaches $\gamma_\mathfrak{l} \leq 0.97$ within the simulation time.}
\label{fig:control}
\end{figure}

\subsection{Spin lattice with three-body interactions}

We considered the following optical $1/2$-spin lattice Hamiltonian
with three-body interactions based on the expression found in
\cite{pachosThreeSpinInteractionsOptical2004},

\begin{equation}
\label{eq:spinham}
\begin{split}
H =& \sum_{i=1}^{N} \omega_i J_i^z + \\
& \sum_{i=1}^{N-1} (g_{i, i+1}^{zz} J_i^z J_{i, i+1}^z +
g_{i, i+1}^{xx} (J_i^x J_{i+1}^x + J_i^y J_{i+1}^y)) + \\
& \sum_{i=1}^{N-2} (g_{i, i+1, i+2}^{zzz} J_i^z J_{i+1}^z J_{i+2}^z +
g_{i, i+1, i+2}^{xzx} (J_i^x J_{i+1}^z J_{i+2}^x + J_i^y J_{i+1}^z J_{i+2}^y))
\end{split}
\end{equation}
where $N=5$ and the following parameters (randomized, while maintaining some separation between different interaction lengths) are used:
$\omega_1 = \omega_2 = 1.13$, $\omega_3 = 1.55$, $\omega_4 = \omega_5 = 2.51$,
$g_{1,2}^{zz} = g_{2,3}^{zz} = 0.77$, $g_{3,4}^{zz} = g_{4,5}^{zz} = 0.46$,
$g_{1,2}^{xx} = g_{2,3}^{xx} = 0.21$, $g_{3,4}^{xx} = g_{4,5}^{xx} = 0.19$,
$g_{1,2,3}^{zzz} = g_{2,3,4}^{zzz} = g_{4,5,6}^{zzz} = 0.08$ and $g_{1,2,3}^{xzx} = g_{2,3,4}^{xzx} = g_{3,4,5}^{xzx} = 0.06$.
Furthermore, $J_i^x$, $J_i^y$, and $J_i^x$ are given by
\eqref{eq:Jx}, \eqref{eq:Jy} and \eqref{eq:Jz}, respectively.

Assume that spins $1$, $2$ and $3$ are fully controllable, whereas spins $4$ and $5$ are not.
Therefore, we let ${\cal H}_\mathfrak{e} \simeq \mathbb{C}^4$ correspond to the last two
spins. Then, since the third spin has direct interactions with ${\cal H}_\mathfrak{e}$,
while the first two interact only with ${\cal H}_\mathfrak{e}$ through the three-body
terms, we let ${\cal H}_\mathfrak{w} \simeq \mathbb{C}^2$ correspond to the third spin
and ${\cal H}_\mathfrak{l} \simeq \mathbb{C}^4$ to the first two. In other words, we let
$n_\mathfrak{l} = n_\mathfrak{e} = 4$,  $n_\mathfrak{w} = 2$ and $U_\mathfrak{s} = \one_\mathfrak{s}$.
We call this choice of frame the $\one$-frame.

The three-body interaction terms in \eqref{eq:spinham} make it so that $\Delta \neq 0$.
Therefore, if we apply the optimization procedure described in Section \ref{sec:subsystem},
we would expect to find a decomposition that better separates the logical subsystem from the environment.
Here we have chosen the regularization hyperparameter to be $\eta = 0.01$.
After finding an optimal $U = \hat{U}$, we call the new frame the $\hat{U}$-frame.
Table \ref{tab:cost_Us} shows the value of the cost function for the original
$\one$-frame and the optimized $\hat{U}$-frame.

\begin{table}[!htpb]
\centering
\begin{tabular}{c|c|c}
\hline
    & $J(U_\mathfrak{s})$ & $J_{\rm reg}(U_\mathfrak{s})$  \\
\hline
$\hat{U}$ & 0.00655  & 0.02005 \\
$\one$    & 0.00680  & 0.02049 \\
\hline
\end{tabular}
\caption{Value of the unregularized and regularized cost functions, i.e.  $J(U_\mathfrak{s})$ and $J_{\rm reg}(U_\mathfrak{s})$ with Hamiltonian \eqref{eq:spinham}
for different spin orders $s$ in the original $\one$-frame and the optimized $\hat{U}$-frame.}
\label{tab:cost_Us}
\end{table}

After choosing the decomposition of ${\cal H}_\mathfrak{s}$, we next choose
an initial state for the ${\cal H}_\mathfrak{w}$ subsystem.
To this end, we compute the OSD decomposition of the Hamiltonian interaction between
${\cal H}_\mathfrak{l}$ and ${\cal H}_\mathfrak{w}$. In this step, it is important to confirm
that the highest singular value is nondegenerate, as can be seen in Table \ref{tab:singvals1}.
Since this is the case, we may proceed to solve the optimization problem
\eqref{eq:optprob_state_full} using an eigenstate of $D_1$ for initialization.

\begin{table}[!htpb]
\centering
\begin{tabular}{c|c|c|c|c}
\hline
    & $s_1$ & $s_2$ & $s_3$ & $s_4$  \\
\hline
$\hat{U}$ & 0.5380 & 0.1548 & 0.1548  & 0.0 \\
$\one$    & 0.5452 & 0.1500 & 0.1500  & 0.0 \\
\hline
\end{tabular}
\caption{Singular values obtained from the OSD decomposition of the $H_\mathfrak{lw}$ part of Hamiltonian \eqref{eq:spinham} in the original $\one$-frame and the optimized $\hat{U}$-frame.}
\label{tab:singvals1}
\end{table}

Figure \ref{fig:Hs} shows the purity dynamics, without control, for different choices
in the subsystem decomposition of $\cal{H}_\mathfrak{s}$ and for different choices
of initial wall state. In particular, we compare the system in the original frame,
where the first spin is taken to be $\cal{H}_\mathfrak{l}$ and the second to be
$\cal{H}_\mathfrak{w}$ and a frame given by the optimized unitary $\hat{U}$.
Then, in each of these frames, we compare the optimized wall state which solves
problem \eqref{eq:optprob_state_full} with a random choice of wall state.

\begin{figure}[ht]
\centering
\includegraphics[width=0.7\linewidth]{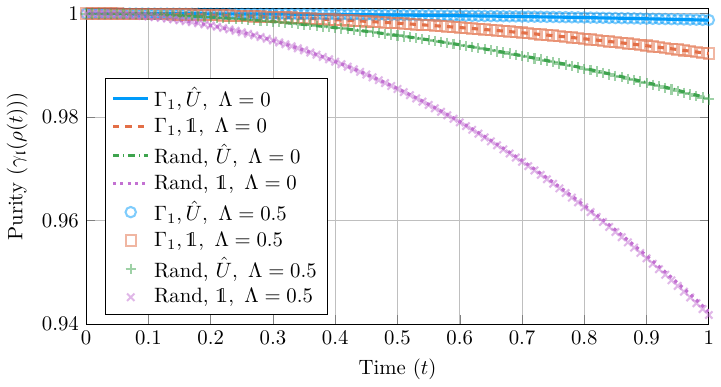}
\caption{%
Purity dynamics of the spin lattice model with different choices of initial wall state and subsystem.
The blue solid line shows the purity with the optimized subsystem and state, whereas
the the red dashed line shows the purity where only the wall state has been optimized.
The green dash-dotted and purple dotted lines show the purity with a random wall state in the optimized and unoptimized subsystems, respectively.
The markers of matching colors show the purity of the system when a Lindbladian pumping term \eqref{eq:pumping} is added to the environment.
}
\label{fig:Hs}
\end{figure}

In Figure \ref{fig:hs} we can see the effect of the different control actions.
These were applied to the system in the $\hat{U}$-frame, i.e. for the optimized $U = \hat{U}$,
and with the wall initialized in the corresponding optimal wall state.
The main difference that we notice with respect
to the Ising chain is that the purity dynamics do not converge to constant. Instead, the slowdown provided by the control
saturates. This is due to the non-zero direct coupling between the logical system and the environment.

\begin{figure}[ht]
\centering
\subfloat[][Repeated measurements.]{
\includegraphics[width=0.45\linewidth]{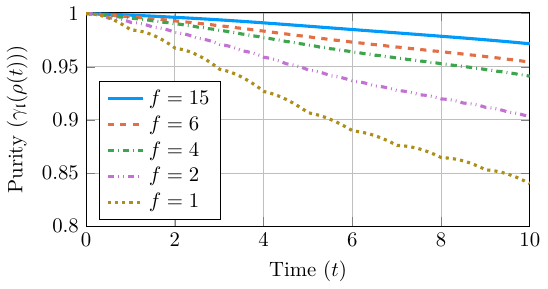}
\label{fig:meas_gains_hs}
}\
\subfloat[][Dissipation.]{
\includegraphics[width=0.45\linewidth]{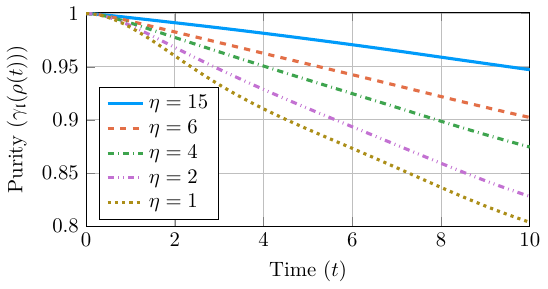}
\label{fig:diss_gains_hs}
}\
\subfloat[][Hamiltonian driving.]{
\includegraphics[width=0.45\linewidth]{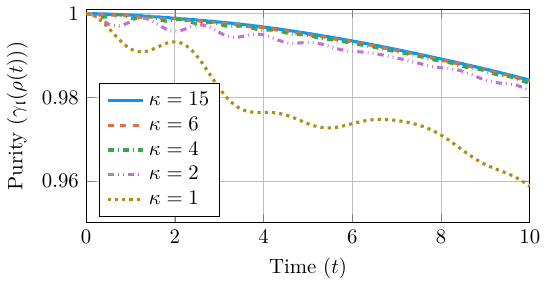}
\label{fig:hamd_gains_hs}
}\
\subfloat[][Time to reach $\gamma_\mathfrak{l} = 0.97$.]{
\includegraphics[width=0.45\linewidth]{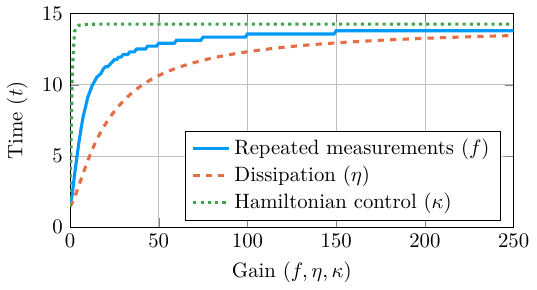}
\label{fig:time_threshd_hs}
}
\caption{\subref{fig:meas_gains_hs}, \subref{fig:diss_gains_hs}, \subref{fig:hamd_gains_hs} Purity dynamics when each of the proposed controls is applied to the spin lattice model for different values of the corresponding control gain variable. \subref{fig:time_threshd_hs} Time to reach purity equal to $0.97$ with different control schemes applied to the spin lattice model.}
\label{fig:hs}
\end{figure}

\subsection{Central spin model}

In the following example, we take $\cal{H}_\mathfrak{s}$ to be a single
$3/2$-spin, which is surrounded by a bath of $1/2$-spins.
The system in question has the following Hamiltonian terms,
\begin{equation}
H_\mathfrak{s} = \omega_s J_\mathfrak{s}^z + \eta_s J_\mathfrak{s}^x,
\end{equation}
\begin{equation}
H_\mathfrak{se} =
A_x J_\mathfrak{s}^x \sum_i J_{\mathfrak{e}_i}^x +
A_y J_\mathfrak{s}^y \sum_i J_{\mathfrak{e}_i}^y +
A_z J_\mathfrak{s}^z \sum_i J_{\mathfrak{e}_i}^z,
\end{equation}
\begin{equation}
H_{e} = \omega_e J_\mathfrak{e}^z + \eta_e J_\mathfrak{e}^x + \lambda_e \sum_{ij} J_{\mathfrak{e}_i}^x J_{\mathfrak{e}_j}^x,
\end{equation}
where
\begin{equation}
J_\mathfrak{s}^x = \left[\begin{matrix} 0 & \sqrt{3}/2 & 0 & 0 \\ \sqrt{3}/2 & 0 & 1 & 0 \\ 0 & 1 & 0 & \sqrt{3}/2 \\ 0 & 0 & \sqrt{3}/2 & 0 \end{matrix}\right] \otimes \one_\mathfrak{e},
\end{equation}
\begin{equation}
J_\mathfrak{s}^y = \left[\begin{matrix} 0 & -\iu \sqrt{3}/2 & 0 & 0 \\ \iu \sqrt{3}/2 & 0 & -\iu & 0 \\ 0 & \iu & 0 & - \iu \sqrt{3}/2 \\ 0 & 0 & \iu \sqrt{3}/2 & 0 \end{matrix}\right] \otimes \one_\mathfrak{e},
\end{equation}
\begin{equation}
J_\mathfrak{s}^z = \left[\begin{matrix} 1.5 & 0 & 0 & 0 \\ 0 & 0.5 & 0 & 0 \\ 0 & 0 & -0.5 & 0 \\ 0 & 0 & 0 & -1.5 \end{matrix}\right] \otimes \one_\mathfrak{e},
\end{equation}
and where $J_\mathfrak{e}^{\{x,y,z\}} = \sum_k J_{\mathfrak{e}_i}^{\{x,y,z\}}$, and $J_{\mathfrak{e}_i}^{\{x,y,z\}}$ are given by
\eqref{eq:Jx}, \eqref{eq:Jy} and \eqref{eq:Jz}, respectively.
In addition, the following parameters were used
$\omega_s = 1.01$, $\eta_s = 0$, $A_x = 0.71$, $A_y = 0$, $A_z = 0.19$,
$\omega_e = 1.92$, $\eta_e = 0$, and $\lambda_e = 0.31$, and the bath was simulated with four $1/2$ spins.

In this system there is no ``obvious'' or ``physical''
decomposition of $\cal{H}_\mathfrak{s}$ into a logical and a wall subsystems.
In the previous example, we had two $1/2$-spins and could therefore choose
a decomposition based on physical intuition, i.e. take each spin as a separate
subsystem.
In this case, a naive approach would be to decompose the Hilbert space $\cal{H}_\mathfrak{s}$ along
virtual degrees of freedom into the tensor product of two $1/2$-spin Hilbert spaces in the basis
that diagonalizes $J_\mathfrak{s}^z$, i.e. $\ket{00} = \ket{3/2}$, $\ket{01} = \ket{1/2}$, $\ket{10} = \ket{-1/2}$
and $\ket{11} = \ket{-3/2}$, where $J_\mathfrak{s}^z \ket{m} = m \ket{m}, m \in \{-3/2, -1/2, 1/2, 3/2\}$.
This is what we refer to in the following as the $\one$-frame. Alternatively, by applying the optimized unitary
change of basis $U = \hat U$ to the $\one$-frame, we obtain the $\hat U$-frame decomposition.

Notice that when $U = \one$, that is, in the $\one$ frame, this decomposition leads
to degenerate dominant singular values when performing the OSD of $H_\mathfrak{sw}$,
violating one of our assumptions for the optimization of the initial wall state.
On the other hand, in the $\hat{U}$ frame, the degeneracy is resolved.

Table \ref{tab:centralspin_cost} shows the values of the cost function for the optimized and original frames. Even though there is a significant reduction
in the coupling of the logical subsystem to the other partitions, we were not able to completely remove it.

\begin{table}[!htpb]
\centering
\begin{tabular}{c|c|c}
\hline
    & $J(U_\mathfrak{s})$ & $J_{\rm reg}(U_\mathfrak{s})$  \\
\hline
$\hat{U}$ & 7.7830  & 7.8728 \\
$\one$    & 30.6792  & 30.6792 \\
\hline
\end{tabular}
\caption{Value of the unregularized and regularized cost functions, i.e.  $J(U_\mathfrak{s})$ and $J_{\rm reg}(U_\mathfrak{s})$ with the central spin model
for different spin orders $s$ in the original $\one$-frame and the optimized $\hat{U}$-frame.}
\label{tab:centralspin_cost}
\end{table}

Table \ref{tab:centralspin_singvals} shows the singular values obtained
from the OSD decomposition of $H_\mathfrak{lw}$. Notice that in the original frame this interaction term
was identically zero, that is why the singular values in the second row of the table are all zero. Indeed, in the original frame
the logical subsystem and the wall were not directly coupled (only indirectly through the environment) and we are able to find a different
subsystem decomposition in which most of the coupling is mediated by the wall ($\|\hat H_\mathfrak{lw}\|^2 = 8.98$ and $\|\hat \Delta\|^2 = 7.78\|$).
In fact, if it was not because of the fact that $\hat \Delta \neq 0$, this system would have an approximate wall state that would be very close to
a perfect one, since $s_1$ is much larger than $s_2$.

\begin{table}[!htpb]
\centering
\begin{tabular}{c|c|c|c|c}
\hline
    & $s_1$ & $s_2$ & $s_3$ & $s_4$  \\
\hline
$\hat{U}$ & 0.9986 & 0.0012 & 0.0 & 0.0 \\
$\one$    & 0.0 & 0.0 & 0.0  & 0.0 \\
\hline
\end{tabular}
\caption{Singular values obtained from the OSD decomposition of the $H_\mathfrak{lw}$ part of the central spin model in the original $\one$-frame and the optimized $\hat{U}$-frame.}
\label{tab:centralspin_singvals}
\end{table}

Figure \ref{fig:Hm} shows the purity dynamics without control for a random initial state on the logical subsystem and
a thermal state with inverse temperature $\beta = 0.01$ in the environment. Although in the Ising spin chain model
the cost functions $\Gamma_1$ and $\Gamma_2$ were identical, due to the fact that $\Delta = 0$, in this case,
when $U = \one$, the opposite is true. The decoherence of the logical system is completely determined
by $\Delta$, as $H_\mathfrak{lw} = 0$. Therefore, neither $\Gamma_2$ nor Proposition \ref{thm:optimwall} can be used,
as $s_1 = s_2 = 0$. However, we are able to find an optimal wall state based only on the indirect LW
coupling captured by $\Gamma_1$.

\begin{figure}[ht]
\centering
\includegraphics[width=0.7\linewidth]{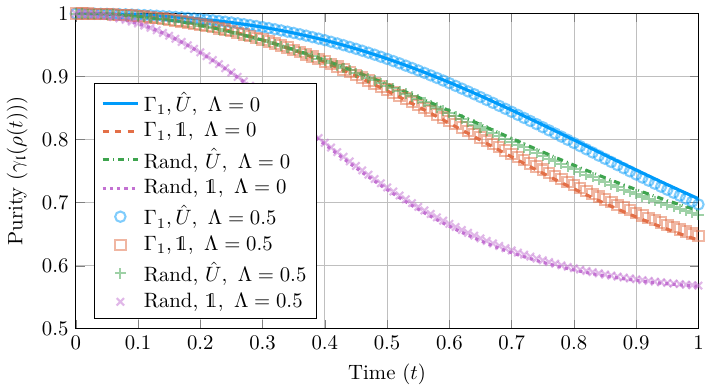}
\caption{%
Purity dynamics of the central spin model with different choices of initial wall state and subsystem. The blue solid line shows the
purity with the optimized subsystem and state, whereas the the red dashed line shows the purity where only the wall state has been optimized. The green dash-dotted and purple dotted lines show the purity with a random wall state in the optimized and unoptimized subsystems, respectively. The markers of matching colors show the purity of the system when a Lindbladian pumping term \eqref{eq:pumping} is added to each of the spins of the bath, showing it has minimal effect on the purity dynamics.
}
\label{fig:Hm}
\end{figure}

As in the case of the spin lattice model, the slow-down provided by the controls saturates at a certain limit determined by
the LE direct coupling. This is evidenced in Figure \ref{fig:hm}.
The controls were applied to the system in the $\hat{U}$-frame, i.e. for the optimized $U = \hat{U}$,
and with the wall initialized in the corresponding optimal wall state.
In this case, we do not observe the peak that was present in
Figure \ref{fig:time_threshd_hs}, as the purity decay in Figure \ref{fig:hamd_gains_hm} is much more straight, without
the pronounced oscillations seen before.

\begin{figure}[ht]
\centering
\subfloat[][Repeated measurements.]{
\includegraphics[width=0.45\linewidth]{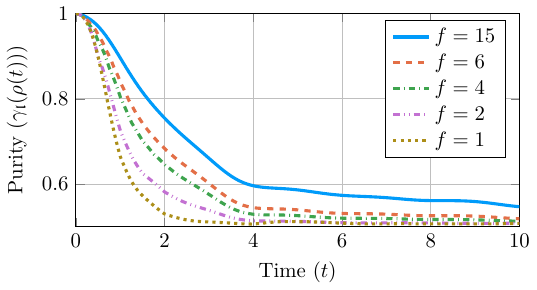}
\label{fig:meas_gains_hm}
}\
\subfloat[][Dissipation.]{
\includegraphics[width=0.45\linewidth]{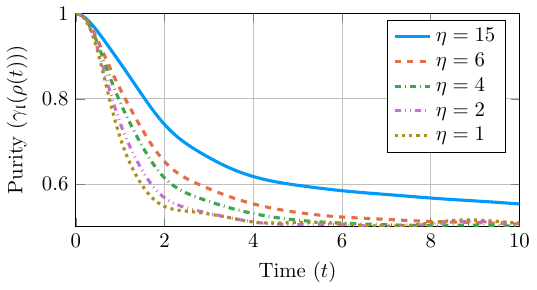}
\label{fig:diss_gains_hm}
}\
\subfloat[][Hamiltonian driving.]{
\includegraphics[width=0.45\linewidth]{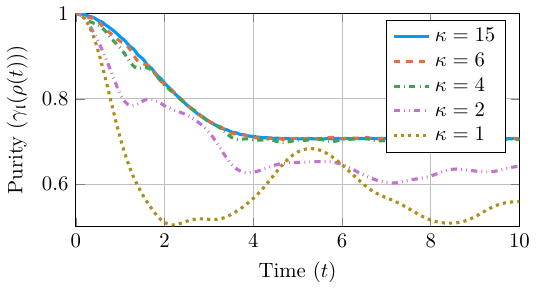}
\label{fig:hamd_gains_hm}
}\
\subfloat[][Time to reach $\gamma_\mathfrak{l} = 0.97$.]{
\includegraphics[width=0.45\linewidth]{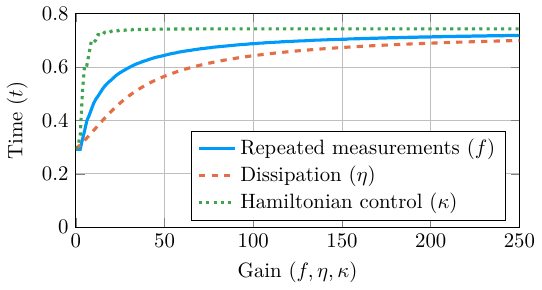}
\label{fig:time_threshd_hm}
}
\caption{\subref{fig:meas_gains_hm}, \subref{fig:diss_gains_hm}, \subref{fig:hamd_gains_hm} Purity dynamics when each of the proposed controls is applied to the central spin model for different values of the corresponding control gain variable. \subref{fig:time_threshd_hm} Time to reach purity equal to $0.97$ with different control schemes applied to the spin lattice model.}
\label{fig:hm}
\end{figure}

\section{Comparison and Integration with Dynamical Decoupling}\label{sec:dd}

A widespread and successful approach for the protection of quantum information is offered by quantum dynamical decoupling (DD) techniques \cite{violaDynamicalDecouplingOpen1999}.
DD consists in actively symmetrizing the dynamics so that the leading orders in the coupling between the system and the environment are suppressed. This is done by the introduction of periodic control such that in the
Floquet (stroboscopic) picture, the system (approximately) follows a unitary evolution.
In this section, we wish to study how our quantum wall method and DD compare and complement each other.

In the following, the target system for the decoupling is considered to be a qubit.
In this case, an idealized universal zero-order DD pulse sequence, that is, a pulse sequence that effectively suppresses any interactions averaged in the pulse cycle time, is given by a sequence of four instantaneous unitary transformations $X,Z,X,Z$ interspaced by free evolution on a given time $T/4.$ Here $T$ represents the length of the control cycle. $X, Z$ are $\pi$-pulses along
the $x$ and $z$ axes, respectively:
\begin{equation}
X = \begin{bmatrix}
0 & 1 \\
1 & 0
\end{bmatrix},
\end{equation}
\begin{equation}
Z = \begin{bmatrix}
1 & 0 \\
0 & -1
\end{bmatrix}.
\end{equation}
The unitary evolution over a control cycle, including the free evolution, reads
\begin{equation}
\label{eq:ddseq}
\begin{split}
U_{\rm DD} =& (Z \otimes \one) e^{-\iu H T/4} (X \otimes \one)
e^{-\iu H T/4} (Z \otimes \one) e^{-\iu H T/4} (X \otimes \one) e^{-\iu H T/4}.
\end{split}
\end{equation}
In some special cases, when the interaction terms are known and are all commuting, decoupling may be achieved by applying pulses along a single
axis. For instance, if the Hamiltonian is a pure dephasing Hamiltonian,
then we do not need to apply the pulses along the $z$ axis, resulting in the
simplified sequence
\begin{equation}
\label{eq:ddseq_simple}
U_{\rm DD} = X e^{-\iu H T/2} X e^{-\iu H T/2}.
\end{equation}
This is an instance of {\em selective decoupling}, where a simpler pulse sequence is tailored to suppress only the dominant error sources, rather than obtaining complete decoupling in some given order \cite{ticozziDynamicalDecouplingQuantum2006}.

The sequence \eqref{eq:ddseq} corresponds to an ideal case in which instantaneous
unbounded $\pi$-pulses are possible. A more realistic setting would be to
consider finite-amplitude pulses. More precisely, we consider the following
time-dependent Hamiltonian,
\begin{equation}
    H_{\rm DD} = H + u_x(t) X + u_z(t) Z,
\end{equation}
where
\begin{equation}
    u_x(t) = \kappa (\Theta((t\ {\rm mod} \frac{T}{2}) - \frac{T}{4} + \tau) - \Theta((t\ {\rm mod} \frac{T}{2}) - \frac{T}{4})),
\end{equation}
\begin{equation}
    u_z(t) = u_x(t - T/4),
\end{equation}
$\Theta(t)$ is the Heaviside function and $\kappa \tau = \pi/2$.
Equivalently, for the simplified sequence \eqref{eq:ddseq_simple},
the corresponding pulses would be
\begin{equation}
    u_x(t) = \kappa (\Theta((t\ {\rm mod} \frac{T}{2}) - \frac{T}{2} + \tau) - \Theta((t\ {\rm mod} \frac{T}{2}) - \frac{T}{2})),
\end{equation}
\begin{equation}
    u_z(t) = 0.
\end{equation}

In the following, we apply the DD to the spin lattice and the central spin models from the previous section
and study how it compares to our quantum wall method with strong Hamiltonian driving, as well as how they
can be used in tandem to further improve the protection of logical information stored in $\mathcal{H}_\mathfrak{l}$.

\subsection{Direct comparison between wall state engineering and selective DD}

Figure \ref{fig:dd} shows the purity dynamics for the spin lattice and the central spin models when DD is applied
to $\mathcal{H}_\mathfrak{l}$.
The DD cycle frequency was set to $f=10$, i.e. $40$ $\pi/2$-pulses per time unit in the case of universal DD
and $20$ $\pi/2$-pulses per unit of time in the case of selective DD.
Finite-amplitude pulses were used, where the control signals were on a fifth of the time, and the system evolved
freely four fifths of the time.
For comparison, the amplitude of the wall Hamiltonian driving was set to $\kappa = 5 \pi$.
This introduces the same amount of energy into the wall subsystem as continuously applying $\pi/2$-pulses of
duration $\tau = 0.1$. Indeed, with these values, we have $\kappa \tau = \pi/2$.
In this sense, this control action is equivalent to
applying $10$ $\pi/2$-pulses per unit of time to the wall subsystem.

In this subsection, we consider the comparison to the cases in which the original frame was used for DD.
The case of DD on the optimized $\hat U$-frame will be considered later in proposing an
``enhanced'' selective DD using our optimized frame. 
For the spin lattice model, both DD approaches, selective and universal, has a performance comparable
to our quantum wall method. This is not the case
for the central spin model, in which universal DD greatly outperforms our method, which has a performance more
in line with selective DD. This is because of a nontrivial residual coupling between the logical system and the environment even in the optimized frame: while the universal DD is able to address that as well, the selective method and the wall state stabilization are expected to struggle. In the next subsection we propose a way to integrate both approaches (universal DD and wall state stabilization) to obtain a protocol that significantly outperforms both.

\begin{figure}[ht]
\centering
\subfloat[][Spin lattice model.]{
\includegraphics[width=0.45\linewidth]{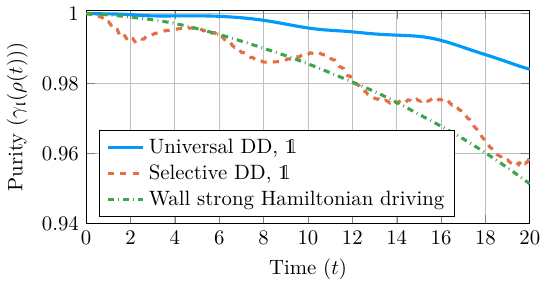}
\label{fig:seldd_Hs}
}\
\subfloat[][Central spin model.]{
\includegraphics[width=0.45\linewidth]{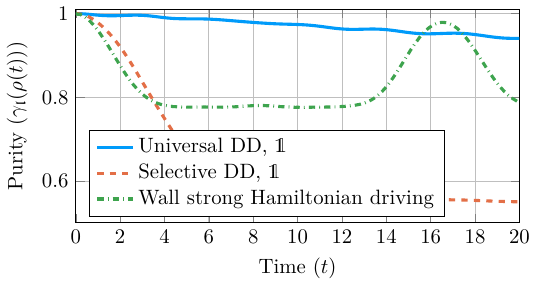}
\label{fig:seldd_Hm}
}
\caption{Comparison between DD and wall state engineering.
The solid blue line represents DD with an universal pulse sequence. The red dashed line represents
selective DD. The green dash-dotted line represents strong Hamiltonian driving on the wall subsystem.
}
\label{fig:dd}
\end{figure}

\subsection{Wall-enhanced DD}

In the following, we consider two levels of integration of our proposed techniques with DD.
The first is to apply DD in the optimized $\hat U$-frame, without additional state stabilization. The second is to simultaneously apply DD on the
logical subsystem and strong Hamiltonian driving on the wall subsystem.

We start by continuing the comparison from the previous subsection, but taking into account DD
applied in the optimized $\hat U$-frame, as shown in Figure \ref{fig:dd_enhc}.
We are mainly interested in selective DD, rather than universal DD.
In the case of the spin lattice model, we observe a slowdown of the
initial purity loss, but it is still faster than that of the universal DD and that of our wall method.
The fact that it still performs worse than universal DD is to be expected, as we know from
Table \ref{tab:singvals1} that the logical subsystem experiences interactions along all three axes of
the Bloch sphere, with relatively uniform intensity. However, in the case of the central spin model,
selective DD provides a practically perfect decoupling when applied in the optimized $\hat U$ frame.
This is compatible with the singular values observed in Table \ref{tab:centralspin_singvals}.

\begin{figure}[ht]
\centering
\subfloat[][Spin lattice model.]{
\includegraphics[width=0.45\linewidth]{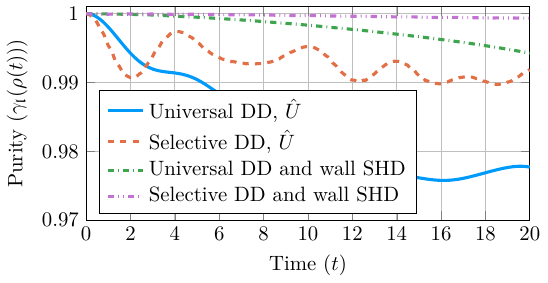}
\label{fig:seldd_enhc_Hs}
}\
\subfloat[][Central spin model.]{
\includegraphics[width=0.45\linewidth]{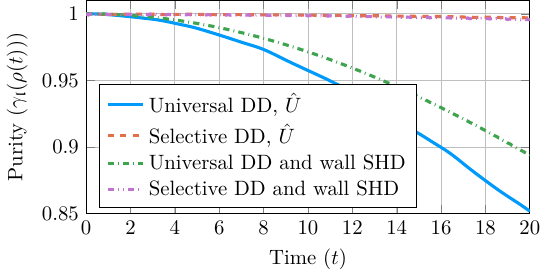}
\label{fig:seldd_enhc_Hm}
}
\caption{Comparison between different kinds of DD, enhanced by wall state engineering.
The solid blue line represents DD with an universal pulse sequence. The red dashed line represents
selective DD. The green dash-dotted line represents universal DD with strong Hamiltonian driving.
The purple dash-dot-dotted line represents selective DD with strong Hamiltonian driving.
}
\label{fig:dd_enhc}
\end{figure}

The next level of enhancement is the application of strong Hamiltonian driving on the wall subsystem
simultaneously to the DD control on the logical subsystem. In Figure \ref{fig:dd_enhc} we see that
it provides practically perfect decoupling in both examples, outperforming all other
decoupling techniques. However, these results are dependent on the combination of DD cycle frequency $f$
and Hamiltonian driving amplitude $\kappa$. For certain combinations of these values, a sort of
anti-Zeno effect \cite{zhangCriterionQuantumZeno2018} may occur, catastrophically accelerating the loss of purity.
Figures \ref{fig:seldd_antizeno_Hs} and \ref{fig:seldd_antizeno_Hm} show examples of the purity dynamics
when this anti-Zeno effect occurs, in contrast to the dynamics when only selective DD is applied and when
it is applied together with Hamiltonian driving with amplitudes both lower and higher than the catastrophic one.
Finally, Figures \ref{fig:seldd_times_Hs} and \ref{fig:seldd_times_Hm} show the time it takes for purity
$\gamma_\mathfrak{l}$ to reach a value of $0.97$ as a function of the Hamiltonian driving amplitude $\kappa$,
both for selective and universal DD. The DD-cycle frequency was kept at $f = 10$.
Although most of these times are not visible in the plots as they exceed the simulation times, the figures are
useful for identifying catastrophic combinations of parameters.
Future work will be devoted to the detailed study of this interference between the two control actions.

\begin{figure}[ht]
\centering
\subfloat[][Sample purity dynamics - Spin lattice.]{
\includegraphics[width=0.45\linewidth]{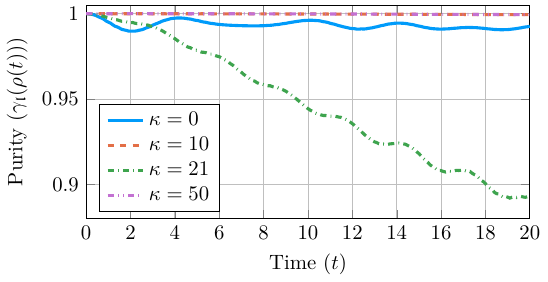}
\label{fig:seldd_antizeno_Hs}
}\
\subfloat[][Sample purity dynamics - Central spin.]{
\includegraphics[width=0.45\linewidth]{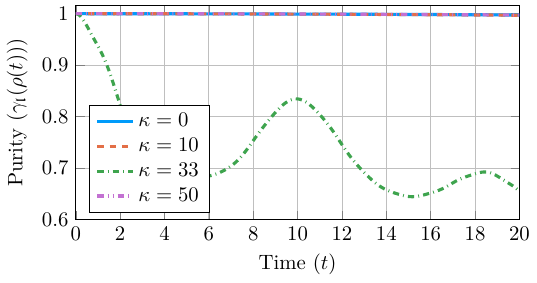}
\label{fig:seldd_antizeno_Hm}
}\
\subfloat[][Time to reach $\gamma_\mathfrak{l} = 0.97$ - Spin lattice.]{
\includegraphics[width=0.45\linewidth]{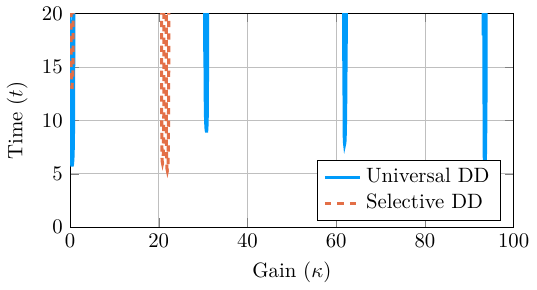}
\label{fig:seldd_times_Hs}
}\
\subfloat[][Time to reach $\gamma_\mathfrak{l} = 0.97$ - Central spin.]{
\includegraphics[width=0.45\linewidth]{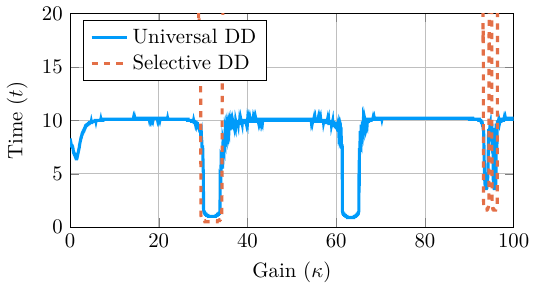}
\label{fig:seldd_times_Hm}
}
\caption{Purity dynamics when selective DD on the logical subsystem
is combined with strong Hamiltonian driving on the wall for different values of the driving amplitude $\kappa$
and a DD cycle frequency of $f = 10$. \subref{fig:seldd_antizeno_Hs}, \subref{fig:seldd_antizeno_Hm} For certain values of
$\kappa$ (here, $\kappa=21$ in \subref{fig:seldd_antizeno_Hs} and $\kappa = 33$ in \subref{fig:seldd_antizeno_Hm}) an ``anti-Zeno'' effect is observed. \subref{fig:seldd_times_Hs}, \subref{fig:seldd_times_Hm} This behavior is
also observed when universal DD is performed, instead of selective DD.}
\label{fig:ddwall}
\end{figure}

\section{Eternally bounded purity}
\label{sec:eternal}

In the numerical exploration of the proposed methods, a remarkable phenomenon emerges, as illustrated in Figure \ref{fig:control_hamd}:  {\em under certain conditions}, the strong Hamiltonian driving exhibits a uniform lower bound on the decrease in purity. Not only that, but such a lower bound can be raised arbitrarily close to one by increasing the driving Hamiltonian gain.  In contrast, all the other stabilization methods for the wall states only induce a reduced decrease rate for purity, which still approaches its minimum of 1/2 at later times. 

We next proceed to characterize the conditions under which we have
the emergence of eternal lower bounds on the purity when strong Hamiltonian control is applied to the quantum wall.
In particular, we find (1) a time-independent lower bound for the purity in function of the initial state and the
system's Hamiltonian (2) a series of conditions on the Hamiltonian so that the purity remains arbitrarily
close to $1$ at every time $t$, if one suitably increases the strong Hamiltonian driving gain $\kappa$.

Consider a system $\cal H = H_\mathfrak{l} \otimes H_\mathfrak{w} \otimes H_\mathfrak{e}$ with Hamiltonian
$H_\kappa = H + \kappa \one_\mathfrak{l} \otimes H_u \otimes \one_\mathfrak{e}$. Let $\ket{\psi^\kappa_i}$ and $\lambda^\kappa_i$
be the eigenstates and eigenvalues of $H_\kappa$, respectively. For the sake of simplicity, in this Section we do not consider any dissipative
term $L_k$: the calculations can be in principle extended to that case, but it would burden the notation and would require working with vectorized representation of the generators.
Therefore, the evolution of the system is purely Hamiltonian,
\begin{equation}
\begin{split}
\rho(t) =& e^{-\iu H_\kappa t} \rho_0 e^{\iu H_\kappa t}
= \Big(\sum_i e^{-\iu \lambda_i t} \ketbra{\psi^\kappa_i}{\psi^\kappa_i}\Big) \rho_0 \Big(\sum_j e^{\iu \lambda_j t} \ketbra{\psi^\kappa_j}{\psi^\kappa_j}\Big).
\end{split}
\end{equation}
Let us define the following instrumental quantities (with the explicit dependence on $\kappa$ is omitted to lighten the notation):
\begin{equation}
\lambda_{abij} = \lambda^\kappa_a - \lambda^\kappa_b + \lambda^\kappa_i - \lambda^\kappa_j,
\end{equation}
\begin{equation}
\rho_{ij} = \bra{\psi^\kappa_i} \rho_0 \ket{\psi^\kappa_j},
\end{equation}
\begin{equation}
\tau_{ij} = \tr_\mathfrak{we}(\ketbra{\psi^\kappa_i}{\psi^\kappa_j}).
\end{equation}
Then, the purity of the logical partition can be written as
\begin{equation}
\begin{split}
\gamma_\mathfrak{l}(\rho(t)) =& \tr(\tr_\mathfrak{we}(\rho(t))^2) \\
=& \tr(\tr_\mathfrak{we}(
(\sum_i e^{-\iu \lambda_i t} \ketbra{\psi^\kappa_i}{\psi^\kappa_i}) \rho_0 (\sum_j e^{\iu \lambda_j t} \ketbra{\psi^\kappa_j}{\psi^\kappa_j}))^2) \\
=& \sum_{abij} e^{-\iu (\lambda_a - \lambda_b + \lambda_i - \lambda_j) t}
\bra{\psi^\kappa_a} \rho_0 \ket{\psi^\kappa_b} \bra{\psi^\kappa_i} \rho_0 \ket{\psi^\kappa_j}
\tr(\tr_\mathfrak{we}(\ketbra{\psi^\kappa_a}{\psi^\kappa_b})
\tr_\mathfrak{we}(\ketbra{\psi^\kappa_i}{\psi^\kappa_j})
) \\
=& \sum_{abij} e^{-\iu \lambda_{abij} t} \rho_{ab} \rho_{ij} \tr(\tau_{ab} \tau_{ij}) \\
=& \sum_{abij} \gamma_{abij}(t),
\end{split}
\end{equation}
where, we defined:
\begin{equation}
\gamma_{abij}(t) = e^{-\iu \lambda_{abij} t} \rho_{ab} \rho_{ij} \tr(\tau_{ab} \tau_{ij}).
\end{equation}

Whenever $\lambda_{abij} = 0$
we have that $\gamma_{abij}(t) = \gamma_{abij}$ is constant. Otherwise,
$\gamma_{abij}(t) + \gamma_{baji}(t) = 2 \cos(\lambda_{abij} t) \rho_{ab} \rho_{ij} \tr(\tau_{ab} \tau_{ij})$.
Let us define the set of indices for which the summands $\gamma_{abij}(t)$ are not constant,
\begin{equation}
{\cal C}_1 = \{(a,b,i,j) | \lambda_{abij} \neq 0\}.
\end{equation}

This allows us to split the expression of the purity into two sums, one consisting only constant terms, and the other
being a superposition of oscillations
\begin{equation}
\label{eq:purity_const}
\begin{split}
\gamma_\mathfrak{l}(\rho(t)) =& \sum_{(a,b,i,j) \in \bar{\mathcal{C}}_1} \rho_{ab} \rho_{ij} \tr(\tau_{ab} \tau_{ij}) +
\sum_{(a,b,i,j) \in \mathcal{C}_1} e^{-\iu \lambda_{abij} t} \rho_{ab} \rho_{ij} \tr(\tau_{ab} \tau_{ij}) \\
=& \bar{\gamma} + \vec{e}(t)^\top \vec{\rho},
\end{split}
\end{equation}
where $\bar{\gamma} = \sum_{(a,b,i,j) \in \bar{\mathcal{C}}_1} \rho_{ab} \rho_{ij} \tr(\tau_{ab} \tau_{ij})$,
$\vec{e}(t) = [e^{-\iu \lambda_{abij} t}]_{(a,b,i,j) \in \mathcal{C}_1}$ and
$\vec{\rho} = [\rho_{ab} \rho_{ij} \tr(\tau_{ab} \tau_{ij})]_{(a,b,i,j) \in \mathcal{C}_1}$, where the last two quantities are vectors that have as many components as the elements of ${\cal C}_1.$
The term $\bar\gamma$ gives us the base value around which the purity oscillates, and the vector $\vec\rho$
encodes the amplitudes of the individual oscillations, collected in $\vec e(t)$.

Ideally, we would like $\vec\rho = 0$.
That would correspond to a perfectly preserved purity, i.e. a unitary evolution of the reduced
state $\rho_\mathfrak{l}$. If there exists a state $\rho_\mathfrak{w} = \ketbra{w}{w}$
such that that is the case, independently of the initial states $\rho_\mathfrak{l}$ and $\rho_\mathfrak{e}$, then it
is a perfect wall state, as the evolution of the reduced state $\rho_\mathfrak{s}$ would be unitary independently
of its initial state and the state of the environment, which is precisely the statement of Definition \ref{def:perfectwall}.
Then, by Proposition \ref{thm:perfectwalldfs} we also know that a DFS exists.

In general, even if $\vec\rho \neq 0$, i.e. the system is not initialized in a DFS,
the bound given in the following proposition applies.

\begin{proposition}
Given a system with purity dynamics as described by \eqref{eq:purity_const},
then the following bound is satisfied for all time $t \in \mathbb{R}$,
\begin{equation}
\gamma_\mathfrak{l}(\rho(t)) \geq \bar{\gamma} - \|\vec{\rho}\|_1.
\end{equation}
\end{proposition}
\begin{proof}
It is easy to see that
\begin{equation*}
\begin{split}
\min_{t \in \mathbb{R}} \gamma_\mathfrak{l}(\rho(t)) =&
\bar{\gamma} + \min_{t \in \mathbb{R}} \vec{e}(t)^\top \vec{\rho} \\
=&
\bar{\gamma} + \min_{t \in \mathbb{R}} \sum_{(a,b,i,j) \in \mathcal{C}_1} e^{-\iu \lambda_{abij} t} \rho_{ab} \rho_{ij} \tr(\tau_{ab} \tau_{ij}) \\
\geq&
\bar{\gamma} + \sum_{(a,b,i,j) \in \mathcal{C}_1} \min_{t \in \mathbb{R}} e^{-\iu \lambda_{abij} t} \rho_{ab} \rho_{ij} \tr(\tau_{ab} \tau_{ij}) \\
=& \bar{\gamma} - \sum_{(a,b,i,j) \in \mathcal{C}_1} |\rho_{ab} \rho_{ij} \tr(\tau_{ab} \tau_{ij})| \\
=& \bar{\gamma} - \|\vec{\rho}\|_1.
\end{split}
\end{equation*}
\end{proof}

At this point, we study the effect of the strong Hamiltonian driving on this bound.
We begin by introducing the following lemma, which will be useful to characterize conditions necessary
for the vector $\vec\rho$ to vanish in the limit $\kappa \to \infty$.

\begin{lemma}
\label{thm:asymp_eigs}
Let $H_\kappa$ be the Hamiltonian of a tripartite system ${\cal H}$ with a tunable parameter $\kappa$, such that
\begin{equation}
H_\kappa = H_0 + \kappa \mathds{1} \otimes H_u \otimes \mathds{1},
\end{equation}
where $H_u$ is non-degenerate. Then, its eigenstates $\ket{\psi^\kappa_j}$ satisfy
\begin{equation}
\label{eq:w_eigs}
\lim_{\kappa \to \infty} \ket{\psi^\kappa_j} =
\sum_{ac} \alpha_{abc}^{(j)} \ket{a} \ket{w_b} \ket{c} = \ket{\tilde \psi_j},
\end{equation}
where $H_u \ket{w_b} = \lambda^u_b \ket{w_b}$. Namely, the eigenstates of $H_\kappa$
converge to eigenstates of $\one \otimes H_u \otimes \one$, where the wall partition
is not entangled to any of the other two partitions.

Finally, the widths of the energy level splittings induced by $H_0$ on $\kappa \one \otimes H_u \otimes \one$ converge to
\begin{equation}
\label{eq:w_eigv}
\lim_{\kappa \to \infty} \bra{\psi_j} (H_\kappa - \kappa \one \otimes H_u \otimes \one) \ket{\psi_j} =
\bra{\tilde \psi_j} H_0 \ket{\tilde \psi_j}.
\end{equation}
\end{lemma}
\begin{proof}

Let $\{\ket{w_b}\}_{b=1}^{n_\mathfrak{w}}$ be the eigenstates of $H_u$ with corresponding eigenvalues
$\{\lambda_b^u\}_{j=1}^{n_\mathfrak{w}}$. The eigenvalues are all distinct, since we are assuming
$H_u$ to be non-degenerate.

Cosider an arbitrary state $\ket{\psi}$,
then
\begin{equation}
\begin{split}
    \lim_{\kappa \to \infty} \frac{H \ket{\psi}}{\|H \ket{\psi}\|} =&
    \lim_{\kappa \to \infty}
    \frac{H_0 \ket{\psi} + \kappa \one \otimes H_u \otimes \one \ket{\psi}}{\|H_0 \ket{\psi} + \kappa \one \otimes H_u \otimes \one \ket{\psi}\|} \\
    =& \frac{\one \otimes H_u \otimes \one \ket{\psi}}{\| \one \otimes H_u \otimes \one \ket{\psi}\|}.
\end{split}
\end{equation}
This proves that indeed the eigenstates converge to those of $\one \otimes H_u \otimes \one$.
As $H_u$ is nondegenerate, its eigenvectors $\{\ket{w_b}\}$ forms a basis of ${\cal H}_\mathfrak{w}$. Let $\{\ket{a}\}_{a=1}^{n_\mathfrak{l}}$
and $\{\ket{c}\}_{c=1}^{n_\mathfrak{e}}$ be bases of ${\cal H}_\mathfrak{l}$ and ${\cal H}_\mathfrak{e}$, respectively.
Now express the state $\ket{\psi}$ in terms of the basis of ${\cal H}$ formed from tensor products of the former,
i.e. $\ket{\psi} = \sum_{abc} \alpha_{abc} \ket{a}\ket{w_b}\ket{c}$,
\begin{equation}
\begin{split}
    \lim_{\kappa \to \infty} \frac{H \ket{\psi}}{\|H \ket{\psi}\|} =&
    \frac{\one \otimes H_u \otimes \one (\sum_{abc} \alpha_{abc} \ket{a}\ket{w_b}\ket{c})}{\| \one \otimes H_u \otimes \one (\sum_{abc} \alpha_{abc} \ket{a}\ket{w_b}\ket{c})\|} \\
    =& \frac{\sum_{abc} \alpha_{abc} \lambda_b^u \ket{a}\ket{w_b}\ket{c}}{\| \sum_{abc} \alpha_{abc} \lambda_b^u \ket{a}\ket{w_b}\ket{c} \|}.
\end{split}
\end{equation}
But since the eigenvalues $\lambda_j^u$ are all distinct, $\ket{\psi}$ is an eigenstate
of $H$ (in the limit $\kappa \to \infty$) if and only if  the coefficients $\alpha_{abc}$
are non-zero only for a single value of the index $b$, i.e. they have the form given by \eqref{eq:w_eigs}.
Finally, \eqref{eq:w_eigv} can be verified by direct calculation.
\end{proof}

We split our study of $\vec\rho$ into a characterization of the coefficients $\rho_{ab} \rho_{ij}$ and the traces $\tr(\tau_{ab} \tau_{ij})$.
We begin by characterizing the cases for which
$\lim_{\kappa \to \infty} \rho_{ab} \rho_{ij} = 0$.
The terms $\rho_{ij}$ involve the
initial state $\rho_0 = \rho_\mathfrak{l} \otimes \ketbra{\hat{w}}{\hat{w}} \otimes \rho_\mathfrak{e}$,
where $\ket{\hat{w}}$ is an eigenstate of $H_u$. Then,
by Lemma \ref{thm:asymp_eigs}, we have that
\begin{equation}
\begin{split}
\lim_{\kappa \to \infty} \rho_{ij} =& \sum_{kl} \overline{\alpha_{k_i l_i}} \alpha_{k_j l_j}
\bra{k_i}\rho_\mathfrak{l}\ket{k_j} \braket{w_i}{\hat{w}}\braket{\hat{w}}{w_j} \bra{l_i}\rho_\mathfrak{e}\ket{l_j},
\end{split}
\end{equation}
which is zero whenever $\ket{w_i} \neq \ket{\hat{w}}$ or $\ket{w_j} \neq \ket{\hat{w}}$.
The following is the set of indices for which the product $\rho_{ab} \rho_{ij}$ cannot be made
arbitrarily small by increasing $\kappa$,
\begin{equation}
{\cal K}_1 = \{(a,b,i,j) | \ket{w_a} = \ket{w_b} = \ket{w_i} =\ket{w_j} = \ket{\hat{w}}\}
\end{equation}

\noindent Next, the trace $\tr(\tau_{ab} \tau_{ij})$ remains to be studied, where
$\tau_{ij}$ is purely dependent on the system's
Hamiltonian.
\begin{equation}
\begin{split}
\lim_{\kappa \to \infty} \tr(\tau_{ab} \tau_{ij}) =&
\sum_{k_a k_b l_a l_b p_i p_j q_i q_j} \alpha_{k_a l_a} \overline{\alpha_{k_b l_b}}
\alpha_{p_i q_i} \overline{\alpha_{p_j q_j}}
\braket{k_b}{p_i} \braket{p_j}{k_a}
\braket{w_a}{w_b} \braket{w_i}{w_j} \braket{l_a}{l_b} \braket{q_i}{q_j} \\
=& \sum_{k_a k_b l_a l_b p_i p_j q_i q_j} \alpha_{k_a l_a} \overline{\alpha_{k_b l_b}}
\alpha_{p_i q_i} \overline{\alpha_{p_j q_j}}
(\bra{k_b} \otimes \bra{p_j}) (\ket{p_i} \otimes \ket{k_a}) \\
&\qquad (\bra{w_b}\bra{l_b} \otimes \bra{w_j}\bra{q_j})
(\ket{w_a}\ket{l_a} \otimes \ket{w_i}\ket{q_i}) \\
=& (\bra{\tilde \psi_b} \otimes \bra{\tilde \psi_j}) {\rm SWAP}_\mathfrak{ll}
(\ket{\tilde \psi_a} \otimes \ket{\tilde \psi_i}),
\end{split}
\end{equation}
where ${\rm SWAP}_\mathfrak{ll} (\ket{a}\ket{b}\ket{c} \otimes \ket{i}\ket{j}\ket{k}) =
\ket{i}\ket{b}\ket{c} \otimes \ket{a}\ket{j}\ket{k}$.
We notice that the trace converges to zero whenever 
${\rm SWAP}_\mathfrak{ll} (\ket{\tilde \psi_a} \otimes \ket{\tilde \psi_i})$ is orthogonal to
$\ket{\tilde \psi_b} \otimes \ket{\tilde \psi_j}$. We now define the set of the indices
for which the trace $\tr(\tau_{ab} \tau_{ij})$ cannot be made arbitrarily small by
increasing $\kappa$,
\begin{equation}
{\cal K}_2 = \{(a,b,i,j) | (\bra{\tilde \psi_b} \otimes \bra{\tilde \psi_j}) {\rm SWAP}_\mathfrak{ll}
(\ket{\tilde \psi_a} \otimes \ket{\tilde \psi_i}) \neq 0\}.
\end{equation}

\begin{proposition}
\label{thm:boundcond}
Given a system with initial state $\rho_0 = \rho_\mathfrak{l} \otimes \ketbra{\hat w}{\hat w} \otimes \rho_\mathfrak{e}$
such that $\gamma_\mathfrak{l}(0) = \gamma_0$,
if ${\cal C}_1 \cap {\cal K}_1 \cap {\cal K}_2 = \varnothing$,
then there exists a $\kappa$ such that the purity
$\gamma_\mathfrak{l}(t) \in (\gamma_0 + \epsilon, \gamma_0 - \epsilon)$ for all time $t$, for any initial
$\rho_\mathfrak{l} \in \mathfrak{D}(\mathcal{H}_\mathfrak{l})$
and $\rho_\mathfrak{e} \in \mathfrak{D}(\mathcal{H}_\mathfrak{e})$,
and for every $\epsilon>0$.
\end{proposition}
\begin{proof}
Let us split $\gamma_\mathfrak{l}(t)$ into the following three sums
\begin{equation}
\begin{split}
\gamma_\mathfrak{l}(t) = \sum_{(a,b,i,j) \in \bar{\cal C}} &e^{-\iu \lambda_{abij} t}
\rho_{ab} \rho_{ij} \tr(\tau_{ab} \tau_{ij}) + \\
\sum_{(a,b,i,j) \in {\cal C}_1 \cap (\bar{\cal K}_1 \cup \bar{\cal K}_2)} &e^{-\iu \lambda_{abij} t}
\rho_{ab} \rho_{ij} \tr(\tau_{ab} \tau_{ij}) + \\
\sum_{(a,b,i,j) \in {\cal C}_1 \cap {\cal K}_1 \cap {\cal K}_2} &e^{-\iu \lambda_{abij} t}
\rho_{ab} \rho_{ij} \tr(\tau_{ab} \tau_{ij})
\end{split}
\end{equation}

If ${\cal C}_1 \cap {\cal K}_1 \cap {\cal K}_2 = \varnothing$, the third sum is zero, as it would
contain no summand. We also know that the summands of the second sum all tend to zero as $\kappa \to \infty$
and that the first sum is composed only of constant summands in time. Therefore, if ${\cal C}_1 \cap {\cal K}_1 \cap {\cal K}_2 = \varnothing$, we can make the purity $\gamma_\mathcal{l}(t)$ remain arbitrarily close to its initial value $\gamma_0$
for all time $t$, by increasing $\kappa$.

\end{proof}

While the condition defining $\mathcal{K}_2$ may seem complex or artificial, if we restrict the asymptotic eigenvalues
$\ket{\tilde \psi_j}$ of the driven Hamiltonian $H_\kappa$ to have a particular product state structure, a simple spectral
interpretation of Proposition \ref{thm:boundcond} emerges. This result is given in the following corollary.

\begin{corollary}
\label{thm:boundcond_noent}
If there exists a choice of bases $\{\ket{j_\mathfrak{l}}\}_{j_\mathfrak{l}=1}^{n_\mathfrak{l}}$ and
$\{\ket{j_\mathfrak{e}}\}_{j_\mathfrak{e}=1}^{n_\mathfrak{e}}$ of
${\cal H}_\mathfrak{l}$ and ${\cal H}_\mathfrak{e}$, respectively, such that
asymptotic eigenstates $\ket{\tilde \psi_j}$ can be written as product states of
the form $\ket{j_\mathfrak{l}}\ket{\hat w}\ket{j_\mathfrak{e}}$, then
${\cal C}_1 \cap {\cal K}_1 \cap {\cal K}_2 = \varnothing$ if and only if
$\lambda_{[j_{\mathfrak{l}}, \hat w, b_{\mathfrak{e}}]} - \lambda_{[i_{\mathfrak{l}}, \hat w, b_{\mathfrak{e}}]} =
\lambda_{[j_{\mathfrak{l}}, \hat w, j_{\mathfrak{e}}]} - \lambda_{[i_{\mathfrak{l}}, \hat w, j_{\mathfrak{e}}]}$
for all $i_{\mathfrak{l}},j_{\mathfrak{l}},b_{\mathfrak{e}},j_{\mathfrak{e}}$, where
$\lambda_{[j_{\mathfrak{l}}, \hat w, j_{\mathfrak{e}}]}$
denotes the eigenvalue of $H_\kappa$ corresponding to the asymptotic eigenstate
$\ket{j_\mathfrak{l}}\ket{\hat w}\ket{j_\mathfrak{e}}$ and $\ket{\hat w}$ is the initial wall state.
\end{corollary}
\begin{proof}

Let $\ket{\psi^{\hat w}_{j_\mathfrak{l}, j_\mathfrak{e}}} = \ket{j_\mathfrak{l}} \ket{\hat w} \ket{j_\mathfrak{e}}$,
then
\begin{equation}
\begin{split}
(\bra{\psi^{\hat w}_{b_\mathfrak{l}, b_\mathfrak{e}}} \otimes
\bra{\psi^{\hat w}_{j_\mathfrak{l}, j_\mathfrak{e}}}) {\rm SWAP}_\mathfrak{ll}
& (\ket{\psi^{\hat w}_{a_\mathfrak{l}, a_\mathfrak{e}}} \otimes
\ket{\psi^{\hat w}_{i_\mathfrak{l}, i_\mathfrak{e}}}) =
\braket{\psi^{\hat w}_{b_\mathfrak{l}, b_\mathfrak{e}}}{\psi^{\hat w}_{i_\mathfrak{l}, a_\mathfrak{e}}}
\braket{\psi^{\hat w}_{j_\mathfrak{l}, j_\mathfrak{e}}}{\psi^{\hat w}_{a_\mathfrak{l}, i_\mathfrak{e}}}
\end{split}
\end{equation}
i.e. $
{\cal K}_1 \cap {\cal K}_2 = \{(a,b,i,j) | b_\mathfrak{l} = i_\mathfrak{l}, a_\mathfrak{l} = j_\mathfrak{l},
b_\mathfrak{e} = a_\mathfrak{e}, j_\mathfrak{e} = i_\mathfrak{e},
a_w = b_w = i_w = j_w = \hat w\}
$.

This allows us to write
\begin{equation}
\begin{split}
{\cal C} \cap {\cal K}_1 \cap {\cal K}_2 =& \{(a,b,i,j) |\
\lambda_{[j_{\mathfrak{l}}, \hat w, b_{\mathfrak{e}}]} - \lambda_{[i_{\mathfrak{l}}, \hat w, b_{\mathfrak{e}}]}
+ \lambda_{[i_{\mathfrak{l}}, \hat w, j_{\mathfrak{e}}]} - \lambda_{[j_{\mathfrak{l}}, \hat w, j_{\mathfrak{e}}]} \neq 0\} \\
=& \{(a,b,i,j) |\
\lambda_{[j_{\mathfrak{l}}, \hat w, b_{\mathfrak{e}}]} - \lambda_{[i_{\mathfrak{l}}, \hat w, b_{\mathfrak{e}}]} \neq
\lambda_{[j_{\mathfrak{l}}, \hat w, j_{\mathfrak{e}}]} - \lambda_{[i_{\mathfrak{l}}, \hat w, j_{\mathfrak{e}}]} \}
\end{split}
\end{equation}
\end{proof}

We finish this section by giving an example of a system where no perfect wall state exists, but the purity
may still be eternally bounded by means of strong Hamiltonian control on the wall.

\subsection{Analytical and numerical example}

Consider a system with the following three-qubit Hamiltonian, where we identify $\mathcal{H}_\mathfrak{l}$ as the space
of the first qubit, $\mathcal{H}_\mathfrak{w}$ of the second and $\mathcal{H}_\mathfrak{e}$ of the third.
\begin{equation}
H_0 = J_1^z + J_2^z + J_3^z + J_1^z J_2^z + J_2^x J_3^z.
\end{equation}
Then, let the following be the driving Hamiltonian
\begin{equation}
H_u = \sigma_z,
\end{equation}
such that the dynamics of the system are given by $H_\kappa = H_0 + \kappa \one_\mathfrak{l} \otimes H_u \otimes \one_\mathfrak{e}$.

Table \ref{tab:eternbound} shows the eigenstates of $H_\kappa$ with their corresponding eigenenergies and asymptotic eigenstates.
The eigenstates are ordered by the value of the eigenenergies $\lambda_j$ for large $\kappa$. Notice that the
eigenstates are not normalized, for legibility purposes.

\begin{table}[!htpb]
\centering
\begin{tabular}{c|c|c|c}
\hline
  $j$  & $\ket{\psi_j}$ & $\lambda_j$ & $\ket{\tilde \psi_j}$ \\
\hline
1 & $(-1 - 4 \kappa + \sqrt{2} \sqrt{1 + 4 \kappa + 8 \kappa^2}) \ket{101} + \ket{111}$ & $-1 - \frac{1}{2 \sqrt{2}} \sqrt{1 + 4 \kappa + 8 \kappa^2}$ & $\ket{111}$ \\
2 & $(-3 - 4 \kappa + \sqrt{2} \sqrt{5 + 12 \kappa + 8 \kappa^2}) \ket{001} + \ket{011}$ & $-\frac{1}{2 \sqrt{2}} (\sqrt{5 + 12 \kappa + 8 \kappa^2})$ & $\ket{011}$ \\
3 & $(1 + 4 \kappa - \sqrt{2} \sqrt{1 + 4 \kappa + 8 \kappa^2}) \ket{100} + \ket{110}$ & $-\frac{1}{2 \sqrt{2}} (\sqrt{1 + 4 \kappa + 8 \kappa^2})$ & $\ket{110}$ \\
4 & $(3 + 4 \kappa - \sqrt{2} \sqrt{5 + 12 \kappa + 8 \kappa^2}) \ket{000} + \ket{010}$ & $1 - \frac{1}{2 \sqrt{2}} \sqrt{5 + 12 \kappa + 8 \kappa^2}$ & $\ket{010}$ \\
5 & $(-1 - 4 \kappa - \sqrt{2} \sqrt{1 + 4 \kappa + 8 \kappa^2}) \ket{101} + \ket{111}$ & $-1 + \frac{1}{2 \sqrt{2}} \sqrt{1 + 4 \kappa + 8 \kappa^2}$ & $\ket{101}$ \\
6 & $(1 + 4 \kappa + \sqrt{2} \sqrt{1 + 4 \kappa + 8 \kappa^2}) \ket{100} + \ket{110}$ & $\frac{1}{2 \sqrt{2}} (\sqrt{1 + 4 \kappa + 8 \kappa^2})$ & $\ket{100}$ \\
7 & $(-3 - 4 \kappa - \sqrt{2} \sqrt{5 + 12 \kappa + 8 \kappa^2}) \ket{000} + \ket{010}$ & $\frac{1}{2 \sqrt{2}} (\sqrt{5 + 12 \kappa + 8 \kappa^2})$ & $\ket{001}$ \\
8 & $(3 + 4 \kappa + \sqrt{2} \sqrt{5 + 12 \kappa + 8 \kappa^2}) \ket{000} + \ket{010}$ & $1 + \frac{1}{2 \sqrt{2}} \sqrt{5 + 12 \kappa + 8 \kappa^2}$ & $\ket{000}$ \\
\hline
\end{tabular}
\caption{Eigenstates and eigenenergies of the driven Hamiltonian.}
\label{tab:eternbound}
\end{table}

If we choose $\ket{\hat w} = \ket{0}$, then we have that
\begin{equation*}
\mathcal{K}_1 = \{5,6,7,8\}^4.
\end{equation*}

Next, we need to find the indices $(a,b,i,j)$ contained in $\mathcal{K}_1$ for which
${\rm SWAP}_\mathfrak{ll} (\ket{\tilde \psi_a} \otimes \ket{\tilde \psi_i})$ is not orthogonal to
$\ket{\tilde \psi_b} \otimes \ket{\tilde \psi_j}$. First, we compute the action of ${\rm SWAP}_\mathfrak{ll}$
on the tensor products of the asymptotic eigenstates,

\begin{equation*}
\begin{split}
{\rm SWAP}_\mathfrak{ll} (\ket{\tilde \psi_i} \otimes \ket{\tilde \psi_i}) =& \ket{\tilde \psi_i} \otimes \ket{\tilde \psi_i}, \\
{\rm SWAP}_\mathfrak{ll} (\ket{\tilde \psi_i} \otimes \ket{\tilde \psi_j}) =& \ket{\tilde \psi_i} \otimes \ket{\tilde \psi_j}
\text{ if } (i,j) \in \{(5,6), (6,5), (7,8), (8,7)\}, \\
{\rm SWAP}_\mathfrak{ll} (\ket{\tilde \psi_i} \otimes \ket{\tilde \psi_j}) =& \ket{\tilde \psi_j} \otimes \ket{\tilde \psi_i}
 \text{ if } (i,j) \in \{(5,7), (7,5), (6,8), (8,6)\}, \\
{\rm SWAP}_\mathfrak{ll} (\ket{\tilde \psi_5} \otimes \ket{\tilde \psi_8}) =& \ket{\tilde \psi_7} \otimes \ket{\tilde \psi_6}, \\
{\rm SWAP}_\mathfrak{ll} (\ket{\tilde \psi_8} \otimes \ket{\tilde \psi_5}) =& \ket{\tilde \psi_6} \otimes \ket{\tilde \psi_7}, \\
{\rm SWAP}_\mathfrak{ll} (\ket{\tilde \psi_6} \otimes \ket{\tilde \psi_7}) =& \ket{\tilde \psi_8} \otimes \ket{\tilde \psi_5}, \\
{\rm SWAP}_\mathfrak{ll} (\ket{\tilde \psi_7} \otimes \ket{\tilde \psi_6}) =& \ket{\tilde \psi_5} \otimes \ket{\tilde \psi_8}. \\
\end{split}
\end{equation*}

Then, we use these results to find the intersection of $\mathcal{K}_1$ and $\mathcal{K}_2$,

\begin{equation*}
\begin{split}
\mathcal{K}_1 \cap \mathcal{K}_2 =& \{(i,i,i,i), 5 \leq i \leq 8\} \cup \\
& \{(i,i,j,j), (i,j) \in (i,j) \in \{(5,6), (6,5), (7,8), (8,7)\}\} \cup \\
& \{(i,j,j,i), (i,j) \in (i,j) \in \{(5,7), (7,5), (6,8), (8,6)\}\} \cup \\
& \{(5,7,8,6), (8,6,5,7), (6,8,7,5), (7,5,6,8)\}.
\end{split}
\end{equation*}

These are the indices of the terms that do not vanish as $\kappa$ is increased.
The only remaining step is to verify if these indices correspond to constant or oscillatory terms.
Clearly, $\lambda_{iiii} = \lambda_{iijj} = \lambda_{ijji} = 0$ for any quantum system. Therefore, we only to check the value
of $\lambda_{abij}$ for $(a,b,i,j) \in \{(5,7,8,6), (8,6,5,7), (6,8,7,5), (7,5,6,8)\}$.
\begin{equation*}
\begin{split}
\lambda_{5,7,8,6} =& -1 + \frac{1}{2 \sqrt{2}} \sqrt{1 + 4 \kappa + 8 \kappa^2} -
\frac{1}{2 \sqrt{2}} \sqrt{5 + 12 \kappa + 8 \kappa^2} + \\
&1 + \frac{1}{2 \sqrt{2}} \sqrt{5 + 12 \kappa + 8 \kappa^2} -
\frac{1}{2 \sqrt{2}} \sqrt{1 + 4 \kappa + 8 \kappa^2} = 0, \\
\lambda_{8,6,5,7} =& - \lambda_{6,8,7,5} = \lambda_{7,5,6,8} = - \lambda_{5,7,8,6} = 0.
\end{split}
\end{equation*}

Therefore,
\begin{equation*}
\mathcal{C}_1 \cap \mathcal{K}_1 \cap \mathcal{K}_2 = \varnothing,
\end{equation*}
implying that the purity $\gamma_\mathfrak{l}$ can be kept eternally bounded, arbitrarily close to $1$,
with the application of strong Hamiltonian driving on the quantum wall.

Alternatively, we could have applied the results of Corollary \ref{thm:boundcond_noent},
since the asymptotic eigenvectors $\ket{\tilde \psi}$ are of the form $\ket{i} \ket{j} \ket{k}$, where
$\ket{i}, \ket{j}, \ket{k} \in \{\ket{0}, \ket{1}\}$. In order to do so, we need to split the index
$j \in \{5, 6, 7, 8\}$ of $\lambda_j$ into $[j_\mathfrak{l}, \hat w, j_\mathfrak{e}]$, i.e. in
order to apply the results of Corollary \ref{thm:boundcond_noent}, we relabel the eigenvalues
according to the last column of Table \ref{tab:eternbound}, such that
\begin{equation*}
\lambda_{[0,0,0]} = \lambda_8, \quad
\lambda_{[0,0,1]} = \lambda_7, \quad
\lambda_{[1,0,0]} = \lambda_6, \quad
\lambda_{[1,0,1]} = \lambda_5.
\end{equation*}
Then, it suffices to show that $\lambda_{[0,0,0]} - \lambda_{[1,0,0]} = \lambda_{[0,0,1]} - \lambda_{[1,0,1]}$
in order to prove that $\mathcal{C}_1 \cap \mathcal{K}_1 \cap \mathcal{K}_2 = \varnothing$.

Finally, Figure \ref{fig:lowerbound} provides a numerical verification of these results. Figure \ref{fig:boundincrease}
shows a numerical computation of the value of the eternal lower bound $\bar \gamma - \|\vec \rho\|_1$
in function of the driving strength $\kappa$. As expected, it converges to $1$ as $\kappa$ is increased.
The bound was computed for $\rho_\mathfrak{l} = (\ket{0} + \iu \ket{1})(\bra{0} - \iu \bra{1})$, i.e. the $+1$
eigenstate of $\sigma_y$. This initial state was chosen because it is not an eigenstate of $\sigma_x$ nor $\sigma_z$.
As for the environment, we set its initial state $\rho_\mathfrak{e}$
as a thermal state with inverse temperature $\beta = 0.01$ as was done in all the previous simulations.
Then, Figure \ref{fig:boundeddynamics} shows the purity dynamics for three different values of the driving strength and
the same initial states, as well as the corresponding lower bounds. As expected, the amplitudes of
the oscillations shrink as $\kappa$ is increased and the purity stays bounded.

\begin{figure}[ht]
\centering
\subfloat[][Lower bound increase.]{
\includegraphics[width=0.45\linewidth]{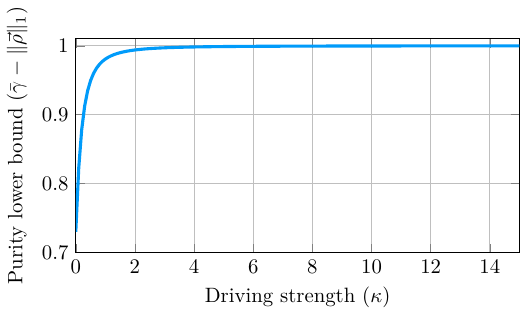}
\label{fig:boundincrease}
}\
\subfloat[][Bounded purity dynamics.]{
\includegraphics[width=0.45\linewidth]{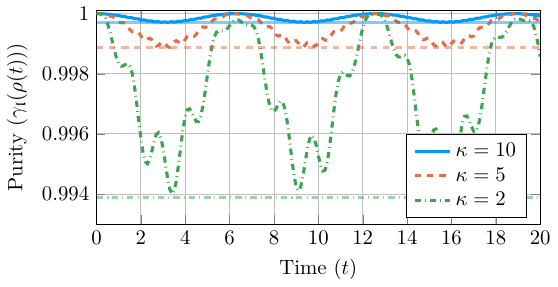}
\label{fig:boundeddynamics}
}
\caption{\subref{fig:boundincrease} Value of the lower bound in function of the driving strength $\kappa$ when
$\rho_\mathfrak{l} = (\ket{0} + \iu \ket{1})(\bra{0} - \iu \bra{1})$, $\ket{\hat w} = \ket{0}$ and $\rho_\mathfrak{e}$
is a thermal state with inverse temperature $\beta = 0.01$.
\subref{fig:boundeddynamics} Purity dynamics for three different values of $\kappa$ and the same initial states.
The horizontal lines of the corresponding color show the value of the time independent lower bound
$\bar \gamma - \|\vec \rho\|_1$ of the purity $\gamma_\mathfrak{l}(t)$.
}
\label{fig:lowerbound}
\end{figure}

\section{Conclusion}

The path towards the realization of robust, scalable devices for quantum information processing should arguably involve the use of multiple techniques aimed to  producing suitable quantum hardware as well as  mitigating the effect of noise in the device via protected encodings, correction protocols and noise suppression control. Being based on the system manipulation, rather than its physical building, multiple noise control methods can in principle be employed in synergy towards the needed noise suppression performance. However, the complexity of the resulting concatenated protocols and their implementation times can limit the resources available for the information processing task. We here focus on developing a method that can seamlessly used in combination with existing one, as the necessary control action does not directly affect the encoded information. 

To this aim, the present work introduces and analyzes the potential of state-stabilization techniques towards the decoupling of a subsystem from environmental interaction.
The proposed framework relies on both analytical and numerical tools, which allow one to search for an (approximate) DFS, whose passive noise isolation properties can be enhanced by actively performing state-stabilizing control on a subsystem that
mediates the dominant interactions with the environment, creating an effective quantum wall state.
One notable advantage with respect to direct error-correcting and DD techniques is that the proposed method employs controls that act only on the wall subsystem and leave the
logical subsystem untouched. This setup allows for simpler implementations of logic operations in the protected subsystem.
It is also suitable for integration with other quantum information protection techniques, such as dynamical decoupling: in this sense, we have demonstrated in one example how our method can be used to enhance the performance of (selective or complete) dynamical decoupling.

When the state-stabilizing control is performed using strong Hamiltonian driving, we found that, under suitable conditions, the purity of the system remains
above a certain threshold {\em for all times}, achieving eternal purity preservation. This remarkable behavior is analyzed  theoretically, and linked to the asymptotic spectrum of the Hamiltonian when the control gain grows unbounded, showing how this affects the oscillatory behavior of the density matrix elements. The sufficient conditions we obtain are tested on a simple three-qubit case, so that we can verify the existence of eternal purity preservation analytically and test the result and bounds numerically.

We believe that the selection and stabilization of wall states in interface subsystems can provide a simple yet effective tool for controlling decoherence. Its possibility of direct integration with other
noise mitigation protocols should make it a desirable technique in the design of new quantum devices. Future developments include refined heuristics in defining the interface subsystems and the dominant Hamiltonian interactions, and testing the approach with more realistic models, as well as with experimental systems.

\section*{Acknowledgments}
F. Ticozzi gratefully acknowledges L. Viola for many valuable discussions on the topics of this work. 
F. Ticozzi and M. Casanova were supported by European Union through
NextGenerationEU, within the National Center for HPC, Big Data and
Quantum Computing under Projects CN00000013, CN 1, and Spoke
10.

\appendix

\section{Riemannian optimization: Gradient descent on manifolds}
\label{sec:riemannianopt}

In this work we have proposed three optimization problems: \eqref{eq:optprob_sub},
\eqref{eq:optprob_state_full} and \eqref{eq:optprob_state}. The first has an optimization
variable that is constrained to belong to ${\rm SU}(n)$,  whereas the latter two
are constrained to $S^{n_w-1}(\mathbb{C})$. Both ${\rm SU}(n)$ and $S^{n_w-1}(\mathbb{C})$ are
smooth manifolds, and in the following we shall denote such a set as $\mathcal{M}$. By equipping them with a Riemannian metric $g_U(X, Y)$, we may
then define the gradient of a functional $J(U): \mathcal{M} \to \mathbb{R}$ as the unique
vector ${\rm grad} J(U) \in T_U \mathcal{M}$ that satisfies
${\rm D} J(U) [\xi] = g_U({\rm grad} J(U), \xi), \forall \xi \in T_U \mathcal{M}$, where ${\rm D} J(U) [\xi]$
is the directional derivative of $J(U)$ along $\xi$ \cite{satoRiemannianOptimizationIts2021}.

To solve an optimization problem $\min_U J(U)$, one could consider the gradient flow
\begin{equation}
\label{eq:gradflow}
\dot U = -{\rm grad} J(U).
\end{equation}
It is easy to see that this flow follows the direction of steepest descent of the functional $J(U)$ and that it
converges to its local minima. The idea behind the gradient descent algorithm is to approximate this flow by discretizing it.
At each point $U$ there is an exponential map ${\rm exp}_U(-{\rm grad}J(U) \epsilon)$ solving the gradient flow \eqref{eq:gradflow},
where ${\rm exp}_U: T_U \mathcal{M} \to \mathcal{M}$ and $\epsilon > 0$ is the step size of the discretization.

In the case of $\mathcal{M} = {\rm SU}(n)$ and $g_U(A,B) = \tr(A^\dagger B)$, we have that
\begin{equation}
{\rm exp}_U(X) = e^{X U^\dagger} U.
\end{equation}
Then, by taking $X = -{\rm grad J(U)} \epsilon$, we recover the expression given in \eqref{eq:expmap_SUn}.

As for $\mathcal{M} = S^{n_w-1}(\mathbb{C})$ and $g_U(a, b) = a^\dagger (\one - \frac{1}{2} U U^\dagger) b$,
we have that
\begin{equation}
{\rm exp}_U(X) = e^{Q X U^\dagger - U X^\dagger Q} U,
\end{equation}
where $Q = \one - \frac{1}{2} U U^\dagger$. Then, by taking $X = -{\rm grad} J(U) \epsilon$, we recover \eqref{eq:expmap_sph}.

The gradient descent algorithm in its simplest form is given by Algorithm \ref{alg:graddesc}. This algorithm is guaranteed
to converge if the step size $\epsilon$ is chosen small enough. A simple addition that can be made to speed up convergence
and to guarantee the monotonic decrease of the cost functional $J(U)$ at each iteration of the while loop is to use
the backtracking algorithm with Armijo-Wolfe conditions to update the step size at each iteration, instead of using
a fixed one. The explanation of these various techniques is beyond the scope of this paper, but further information
can be found in \cite{satoRiemannianOptimizationIts2021}.

\begin{algorithm}
    \caption{Gradient descent}
    \label{alg:graddesc}
    \begin{algorithmic}[1]
        \REQUIRE $U_0 \in \mathcal{M}$
        \ENSURE $U_f \in \mathcal{M}$
        
        \STATE $k \gets 0$
        \STATE $U_k \gets U_0$
        
        \WHILE{ $g_{U_k}({\rm grad} J(U_k), {\rm grad} J(U_k)) > g_{\rm min}$ }
            \STATE $U_{k+1} \gets {\rm exp}_{U_k}(- {\rm grad} J(U_k) \epsilon)$
            \STATE $k \gets k + 1$
        \ENDWHILE

        \STATE $U_f \gets U_k$
    \end{algorithmic}
\end{algorithm}

\section{Proof of Proposition \ref{thm:avg_2nd_deriv}}
\label{sec:avg_2nd_deriv}

We start this section by proving the following lemma on the properties of the uniform distribution of density matrices $\rho$ on the set
of pure states $\mathrm{extr}(\mathfrak{D}(\mathcal{H}_\mathfrak{u}))$.

\begin{lemma}
\label{thm:averages}
Let $\rho_{jk} = \bra{j} \rho \ket{k}$ be the entries of a pure density matrix $\rho \in \mathrm{extr}(\mathfrak{D}(\mathcal{H}_\mathfrak{u}))$, then
\begin{enumerate}
\item $\langle \rho_{jk} \rangle_\mathfrak{u} = \frac{1}{n_\mathfrak{u}} \delta_{jk}$.
\item $\langle \rho_{jk} \rho_{cb} \rangle_\mathfrak{u} = \frac{1}{n_\mathfrak{u} (n_\mathfrak{u}+1)} (\delta_{jk} \delta_{bc} + \delta_{jb} \delta_{kc})$
\end{enumerate}
Where the average is taken over $\rho$ uniformly distributed on the set of pure states $\mathrm{extr}(\mathfrak{D}(\mathcal{H}_\mathfrak{u}))$.
\end{lemma}
\begin{proof}
Let $\rho = \ketbra{\psi}{\psi}$ with $\ket{\psi} = \sum_{i = 0}^{n_\mathfrak{u} - 1} \alpha_i e^{\iu \phi_i} \ket{i}$,
where $\alpha_i \in [0, 1]$, $\phi_i \in [0, 2\pi)$ and $\sum_i \alpha_i^2 = 1$.
Then, let $\ket{\psi}$ be uniformly distributed on the unit complex-valued $(n_\mathfrak{u}-1)$-sphere. That implies that the
vector collecting square of the amplitudes $[\alpha_i^2]_{i=0}^{n_\mathfrak{u}-1}$ is given by the
uniform Dirichlet distribution $\mathrm{Dir}(1/2)$
\cite{guardiolaSphericalDirichletDistribution2020, lanFlexibleBayesianDynamic2020a},
whereas the complex phases $e^{\iu \phi_i}$ are independent from each other and
from the amplitudes, and (identically) uniformly distributed along the unit circle on the complex plane $\mathbb{C}$.

As a consequence of these facts, we know that
\begin{itemize}
\item $\langle \alpha_i^2 \rangle = \frac{1}{n_\mathfrak{u}}$, which is the first moment of $\mathrm{Dir}(1/2)$.
\item $\langle \alpha_i^4 \rangle = \frac{2}{n_\mathfrak{u} (n_\mathfrak{u} + 1)}$, which is the second moment of $\mathrm{Dir}(1/2)$.
\item $\langle e^{\iu \phi_i} \rangle = 0$, which is the center of the unit circle on $\mathbb{C}$.
\item $\langle e^{\iu (\phi_i - \phi_j)} \rangle = \delta_{ij}$, which comes from the independence between $e^{\iu \phi_i}$ and $e^{-\iu \phi_j}$.
\item $\langle \alpha_j \alpha_k \alpha_b \alpha_c e^{\iu (\phi_j - \phi_k - \phi_b + \phi_c)} \rangle = \langle \alpha_j \alpha_k \alpha_b \alpha_c \rangle (\delta_{jk} \delta_{bc} + \delta_{jb} \delta_{kc} - \delta_{jk} \delta_{bc} \delta_{jb} \delta_{kc})$, which comes from the previous point and the independence between amplitudes and phases.
\end{itemize}

Equipped with the previous points, we can know compute our quantities of interest $\langle \rho_{ij} \rangle_\mathfrak{u}$ and
$\langle \rho_{jk} \rho_{cb} \rangle_\mathfrak{u}$. By denoting $\rho_{ij} = \alpha_i \alpha_j e^{\iu (\phi_i - \phi_j)}$, we have that
\begin{equation}
\langle \rho_{ij} \rangle = \langle \alpha_i \alpha_j e^{\iu (\phi_i - \phi_j)} \rangle
= \langle \alpha_i \alpha_j \rangle \delta_{ij}
= \langle \alpha_i^2 \rangle \delta_{ij}
= \frac{1}{n_\mathfrak{u}} \delta_{ij}.
\end{equation}

In order to compute the second quantity of interest, we need to compute $\langle \alpha_i^2 \alpha_j^2\rangle$ for $i \neq j$ first.
This is done by taking advantage of the normalization of the amplitudes,
\begin{equation}
\begin{split}
\langle \alpha_i^4 \rangle =& \langle \alpha_i^2 (1 - \sum_{j \neq i} \alpha_j^2) \rangle \\
=& \langle \alpha_i^2 \rangle - (n_\mathfrak{u} - 1) \langle \alpha_i^2 \alpha_j^2\rangle \\
=& \frac{1}{n_\mathfrak{u}} - (n_\mathfrak{u} - 1) \langle \alpha_i^2 \alpha_j^2\rangle = \frac{2}{n_\mathfrak{u} (n_\mathfrak{u} + 1)}.
\end{split}
\end{equation}
Then, we have that
\begin{equation}
\langle \alpha_i^2 \alpha_j^2 \rangle = \frac{1}{n_\mathfrak{u} (n_\mathfrak{u} + 1)}, \ i \neq j.
\end{equation}

Finally, we are able to compute the second quantity
\begin{equation}
\begin{split}
\langle \rho_{jk} \rho_{cb} \rangle =&
\langle \alpha_j \alpha_k \alpha_b \alpha_c e^{\iu \phi_j - \iu \phi_k - \iu \phi_b + \iu \phi_c} \rangle \\
=& \langle \alpha_j \alpha_k \alpha_b \alpha_c \rangle (\delta_{jk} \delta_{bc} + \delta_{jb} \delta_{kc} - \delta_{jk} \delta_{bc} \delta_{jb} \delta_{kc}) \\
=& \langle \alpha_j^2 \alpha_c^2 \rangle (\delta_{jk} \delta_{bc} + \delta_{jb} \delta_{kc}) (1 - \delta_{jc}) + \langle \alpha_j^4 \rangle \delta_{jk} \delta_{bc} \delta_{jb} \delta_{kc} \\
=& \frac{1}{n_\mathfrak{u} (n_\mathfrak{u} + 1)} (\delta_{jk} \delta_{bc} + \delta_{jb} \delta_{kc}) (1 - \delta_{jc}) + \frac{2}{n_\mathfrak{u} (n_\mathfrak{u} + 1)} \delta_{jk} \delta_{bc} \delta_{jb} \delta_{kc} \\
=& \frac{1}{n_\mathfrak{u} (n_\mathfrak{u} + 1)} (\delta_{jk} \delta_{bc} + \delta_{jb} \delta_{kc}).
\end{split}
\end{equation}
\end{proof}

These results can now be used to compute the values of $\langle \tr(\sigma_i \sigma_a \rho) \rangle$ and
$\langle \tr(\sigma_i \rho \sigma_a \rho) \rangle$, which will be necessary to prove Proposition \ref{thm:avg_2nd_deriv}.
The computation is given in the following lemma.

\begin{lemma}
\label{thm:integral}
Let
$\{\sigma_i\}_{i=0}^{n_\mathfrak{u}^2-1} = \{\frac{1}{\sqrt{n_\mathfrak{u}}} \one\} \cup \{\lambda^s_{jk}\}_{0 \leq k < j \leq n_\mathfrak{u}-1} \cup \{\lambda^a_{jk}\}_{0 \leq j < k \leq {n_\mathfrak{u}}-1} \cup \{\lambda^d_k\}_{1 \leq k \leq n_\mathfrak{u}-1}$,
where $\lambda^s_{jk}$, $\lambda^a_{jk}$ and  $\lambda^d_{jk}$ are the symmetric,
antisymmetric and diagonal generalized $n_\mathfrak{u} \times n_\mathfrak{u}$ Gell-Mann matrices \cite{bertlmannBlochVectorsQudits2008} (normalized to $1$), respectively, i.e.
\begin{equation}
\begin{split}
\lambda^s_{jk} =& \frac{1}{\sqrt{2}} (\ketbra{j}{k} + \ketbra{k}{j}) \\
\lambda^a_{jk} =& \frac{1}{\sqrt{2}} (\iu \ketbra{j}{k} - \iu \ketbra{k}{j}) \\
\lambda^d_{k} =& \frac{1}{\sqrt{k (k+1)}} (\sum_{j=0}^{k-1} \ketbra{j}{j} - k \ketbra{k}{k}).
\end{split}
\end{equation}
Let in particular $\sigma_0 = \frac{1}{\sqrt{n_\mathfrak{u}}} \one$. Then for $i,j > 0$,
\begin{enumerate}
\item $\langle \tr(\sigma_i \sigma_a \rho) \rangle = \frac{1}{n_\mathfrak{u}} \delta_{ia}$ and
\item $\langle \tr(\sigma_i \rho \sigma_a \rho) \rangle = \frac{1}{n_\mathfrak{u} (n_\mathfrak{u} + 1)} \delta_{ia}$,
\end{enumerate}
where the average is taken over all pure states $\rho \in \mathrm{extr}(\mathfrak{D}(\mathcal{H}_\mathfrak{u}))$.
\end{lemma}
\begin{proof}
The first point is obtained as a direct application of Lemma \ref{thm:averages} and of the orthonormality of $\{\sigma_i\}_i$,
\begin{equation}
\langle \tr(\sigma_i \sigma_a \rho) \rangle = \tr(\sigma_i \sigma_a \langle \rho \rangle) = \frac{1}{n_\mathfrak{u}} \tr(\sigma_i \sigma_a \one) = \frac{1}{n_\mathfrak{u}} \delta_{ia}.
\end{equation}

The second point is proven by cases. First, let
$\sigma_i = \lambda^s_{jk}$ and $\sigma_a = \lambda^s_{bc}$, i.e. $j > k$ and $b > c$, then
\begin{equation}
\begin{split}
\langle \tr(\lambda^s_{jk} \rho \lambda^s_{bc} \rho) \rangle =&
\frac{1}{2} \langle \tr((\ketbra{j}{k} + \ketbra{k}{j}) \rho (\ketbra{b}{c} + \ketbra{c}{b}) \rho) \rangle \\
=& \frac{1}{2} (\langle \rho_{kb} \rho_{cj} \rangle + \langle \rho_{jb} \rho_{ck} \rangle + \langle \rho_{kc} \rho_{bj} \rangle + \langle \rho_{jc} \rho_{bk} \rangle).
\end{split}
\end{equation}
By applying Lemma \ref{thm:averages}, we know that the first and last terms of this sum are zero, since we have that $k \neq j$ and $c \neq b$, whereas the middle two are equal, resulting in
\begin{equation}
\langle \tr(\lambda^s_{jk} \rho \lambda^s_{bc} \rho) \rangle = \frac{1}{n_\mathfrak{u} (n_\mathfrak{u} + 1)} \delta_{jb} \delta_{kc}.
\end{equation}

If $\sigma_i = \lambda^a_{jk}$ and $\sigma_a = \lambda^a_{bc}$, i.e. $j < k$ and $b < c$, then
\begin{equation}
\begin{split}
\langle \tr(\lambda^a_{jk} \rho \lambda^a_{bc} \rho) \rangle =&
\frac{1}{2} \langle \tr((\iu \ketbra{j}{k} - \iu \ketbra{k}{j}) \rho (\iu \ketbra{b}{c} - \iu \ketbra{c}{b}) \rho) \rangle \\
=& \frac{1}{2} (- \langle \rho_{kb} \rho_{cj} \rangle + \langle \rho_{jb} \rho_{ck} \rangle + \langle \rho_{kc} \rho_{bj} \rangle - \langle \rho_{jc} \rho_{bk} \rangle).
\end{split}
\end{equation}
The same argument as in the previous case applies here, resulting in
\begin{equation}
\langle \tr(\lambda^a_{jk} \rho \lambda^a_{bc} \rho) \rangle = \frac{1}{n_\mathfrak{u} (n_\mathfrak{u} + 1)} \delta_{jb} \delta_{kc}.
\end{equation}

If $\sigma_i = \lambda^d_{k}$ and $\sigma_a = \lambda^d_{c}$, then
\begin{equation}
\begin{split}
\langle \tr(\lambda^d_{k} \rho \lambda^d_{c} \rho) \rangle
=& \frac{1}{\sqrt{k (k+1)}} \frac{1}{\sqrt{c (c+1)}} \\
& \langle \tr((\sum_{j=0}^{k-1} \ketbra{j}{j} - k \ketbra{k}{k}) \rho (\sum_{b=0}^{c-1} \ketbra{b}{b} - c \ketbra{c}{c}) \rho) \rangle \\
=& \frac{1}{\sqrt{k (k+1)}} \frac{1}{\sqrt{c (c+1)}} \\
& (\sum_{j=0}^{k-1} \sum_{b=0}^{c-1} \langle \rho_{jb} \rho_{bj} \rangle
- k \sum_{b=0}^{c-1} \langle \rho_{kb} \rho_{bk} \rangle
- c \sum_{j=0}^{k-1} \langle \rho_{jc} \rho_{cj} \rangle
+ k c \langle \rho_{kc} \rho_{ck} \rangle) \\
=& \frac{1}{\sqrt{k (k+1)}} \frac{1}{\sqrt{c (c+1)}} \frac{1}{n_\mathfrak{u} (n_\mathfrak{u} + 1)} \\
& (\sum_{j=0}^{k-1} \sum_{b=0}^{c-1} \delta_{jb} - k \sum_{b=0}^{c-1} \delta_{kb} - c \sum_{j=0}^{k-1} \delta_{jc} + k c \delta_{kc}).
\end{split}
\end{equation}

The first term of the sum is equal to $\min\{k, c\}$, whereas the addition of the second and third terms is equal to $-\min\{k, c\} (1 - \delta_{kc})$.
Then, we have that
\begin{equation}
\begin{split}
\langle \tr(\lambda^d_{k} \rho \lambda^d_{c} \rho) \rangle
=& \frac{1}{\sqrt{k (k+1)}} \frac{1}{\sqrt{c (c+1)}} \frac{1}{n_\mathfrak{u} (n_\mathfrak{u} + 1)} \\
& (\min\{k, c\} - \min\{k, c\} (1 - \delta_{kc}) + k c \delta_{kc}) \\
=& \frac{1}{\sqrt{k (k+1)}} \frac{1}{\sqrt{c (c+1)}} \frac{1}{n_\mathfrak{u} (n_\mathfrak{u} + 1)} (\min\{k, c\} + k c) \delta_{kc} \\
=& \frac{1}{k (k+1)} \frac{1}{n_\mathfrak{u} (n_\mathfrak{u} + 1)} (k + k^2) \delta_{kc} \\
=& \frac{1}{n_\mathfrak{u} (n_\mathfrak{u} + 1)} \delta_{kc}
\end{split}
\end{equation}

Now we have to prove that whenever $\sigma_i$ and $\sigma_a$ are of different kinds, e.g. one is $\sigma_i$ symmetric
and $\sigma_a$ is diagonal, the average $\langle \tr(\sigma_i \rho \sigma_a \rho) \rangle$ is zero.
If $\sigma_i = \lambda^s_{jk}$ and $\sigma_a = \lambda^a_{bc}$, i.e. $j > k$ and $b < c$, we have that
\begin{equation}
\begin{split}
\langle \tr(\lambda^s_{jk} \rho \lambda^a_{bc} \rho) \rangle =& \frac{1}{2} \langle \tr((\ketbra{j}{k} + \ketbra{k}{j}) \rho (\iu \ketbra{b}{c} - \iu \ketbra{c}{b}) \rho) \rangle \\
=& \frac{1}{2} (\iu \langle \rho_{kb} \rho_{cj} \rangle + \iu \langle \rho_{jb} \rho_{ck} \rangle - \iu \langle \rho_{kc} \rho_{bj} \rangle - \iu \langle \rho_{jc} \rho_{bk} \rangle).
\end{split}
\end{equation}

From the first case, we know that the first and last terms are zero, whereas the two in the middle were equal. However, now they appear
with opposite signs, making the entire expression equal to zero. Therefore,
\begin{equation}
\langle \tr(\lambda^s_{jk} \rho \lambda^a_{bc} \rho) \rangle = 0.
\end{equation}

If $\sigma_i = \lambda^s_{jk}$ and $\sigma_a = \lambda^d_c$, i.e. $j > k$, we have that
\begin{equation}
\begin{split}
\langle \tr(\lambda^s_{jk} \rho \lambda^d_c \rho) \rangle =& \frac{1}{\sqrt{2 c (c+1)}} \langle \tr((\ketbra{j}{k} + \ketbra{k}{j}) \rho (\sum_{b=0}^{c-1} \ketbra{b}{b} - c \ketbra{c}{c}) \rho) \rangle \\
=& \frac{1}{\sqrt{2 c (c+1)}}
(\sum_{b=0}^{c-1} \langle \rho_{kb} \rho_{bj} \rangle + \sum_{b=0}^{c-1} \langle \rho_{jb} \rho_{bk} \rangle
- c \langle \rho_{kc} \rho_{cj} \rangle - c \langle \rho_{jc} \rho_{ck} \rangle).
\end{split}
\end{equation}

Since $j \neq k$, by Lemma \ref{thm:averages} we know that all four terms inside the parenthesis in the previous expression are equal to zero.
Therefore
\begin{equation}
\begin{split}
\langle \tr(\lambda^s_{jk} \rho \lambda^d_{bc} \rho) \rangle =& 0.
\end{split}
\end{equation}

Finally, if $\sigma_i = \lambda^a_{jk}$ and $\sigma_a = \lambda^d_{c}$, i.e. $j < k$, we have that
\begin{equation}
\begin{split}
\langle \tr(\lambda^a_{jk} \rho \lambda^d_c \rho) \rangle =& \frac{1}{\sqrt{2 c (c+1)}} \langle \tr((\iu \ketbra{j}{k} - \iu \ketbra{k}{j}) \rho (\sum_{b=0}^{c-1} \ketbra{b}{b} - c \ketbra{c}{c}) \rho) \rangle \\
=& \frac{1}{\sqrt{2 c (c+1)}}
(\iu \sum_{b=0}^{c-1} \langle \rho_{kb} \rho_{bj} \rangle - \iu \sum_{b=0}^{c-1} \langle \rho_{jb} \rho_{bk} \rangle
- \iu c \langle \rho_{kc} \rho_{cj} \rangle + \iu c \langle \rho_{jc} \rho_{ck} \rangle),
\end{split}
\end{equation}
which is is zero by the same argument as in the previous case. All other cases are automatically verified by the cyclicity of the trace.
\end{proof}

With these results in hand, we can now proceed to prove Proposition \ref{thm:avg_2nd_deriv}.

\begin{proposition*}
Given a tripartite system with Hilbert space $\cal H = H_\mathfrak{l} \otimes H_\mathfrak{w} \otimes H_\mathfrak{e}$
and factorized initial state $\rho_0 = \rho_\mathfrak{l} \otimes \rho_\mathfrak{w} \otimes \rho_\mathfrak{e}$,
which follows Lindbladian dynamics with arbitrary Hamiltonian and noise operators $L_m = \one_\mathfrak{l} \otimes \one_\mathfrak{w} \otimes L^\mathfrak{e}_m$,
the second derivative of the purity $\gamma_\mathfrak{l}(\rho)$ has the following average initial value over all the reduced states $\rho_\mathfrak{l}$,
\begin{equation}
\langle \ddot\gamma_\mathfrak{l}(\rho)|_{t=0} \rangle_\mathfrak{l} =
4 \nu_n \sum_{\substack{i > 0 \\ bcjk \geq 0}} g_{ibc} g_{ijk}
(\tau_{c0} \tau_{0k} \tr(\sigma_b^\mathfrak{w} \rho_\mathfrak{w}) \tr(\sigma_j^\mathfrak{w} \rho_\mathfrak{w})
- \tau_{ck} \tr(\sigma_b^\mathfrak{w} \sigma_j^\mathfrak{w} \rho_\mathfrak{w})),
\end{equation}
where $\nu_n = \frac{1}{n_\mathfrak{l}} - \frac{1}{n_\mathfrak{l} (n_\mathfrak{l} + 1)}$ and
$\tau_{ck} = \tr(\sigma^\mathfrak{e}_c \sigma^\mathfrak{e}_k \rho_\mathfrak{e})$.
\end{proposition*}
\begin{proof}
If $\rho_0 = \rho_\mathfrak{l} \otimes \rho_\mathfrak{w} \otimes \rho_\mathfrak{e}$, we have that
\begin{equation}
\begin{split}
\tr(\tr_\mathfrak{we}([\sigma_{abc}, [\sigma_{ijk}, \rho_0]]) \rho_\mathfrak{l}) =&
\tr(\sigma_b^\mathfrak{w} \sigma_j^\mathfrak{w} \rho_\mathfrak{w}) \tr(\sigma_c^\mathfrak{e} \sigma_k^\mathfrak{e} \rho_\mathfrak{e})
\tr(\sigma_a^\mathfrak{l} \sigma_i^\mathfrak{l} \rho_\mathfrak{l}^2 - \sigma_a^\mathfrak{l} \rho_\mathfrak{l} \sigma_i^\mathfrak{l} \rho_\mathfrak{l}) + \\
& \tr(\sigma_j^\mathfrak{w} \sigma_b^\mathfrak{w} \rho_\mathfrak{w}) \tr(\sigma_k^\mathfrak{e} \sigma_c^\mathfrak{e} \rho_\mathfrak{e})
\tr(\sigma_i^\mathfrak{l} \sigma_a^\mathfrak{l} \rho_\mathfrak{l}^2 - \sigma_i^\mathfrak{l} \rho_\mathfrak{l} \sigma_a^\mathfrak{l} \rho_\mathfrak{l}),
\end{split}
\end{equation}
and that
\begin{equation}
\begin{split}
\tr(\tr_\mathfrak{we}([\sigma_{abc}, \rho_0]) \tr_\mathfrak{we}([\sigma_{ijk}, \rho_0])) =&
\tr(\sigma_b^\mathfrak{w} \rho_\mathfrak{w}) \tr(\sigma_j^\mathfrak{w} \rho_\mathfrak{w})
\tr(\sigma_c^\mathfrak{e} \rho_\mathfrak{e}) \tr(\sigma_k^\mathfrak{e} \rho_\mathfrak{e}) \\
&\qquad \tr(2 \sigma_a^\mathfrak{l} \rho_\mathfrak{l} \sigma_i^\mathfrak{l} \rho_\mathfrak{l} - \sigma_a^\mathfrak{l} \sigma_i^\mathfrak{l} \rho_\mathfrak{l}^2 - \sigma_i^\mathfrak{l} \sigma_a^\mathfrak{l} \rho_\mathfrak{l}^2).
\end{split}
\end{equation}
Furthermore, if we also have
$L_m = \one_\mathfrak{l} \otimes \one_\mathfrak{w} \otimes L_m^\mathfrak{e}$,
then
\begin{equation}
\tr_\mathfrak{we}({\cal D}_{L_n} \circ {\cal D}_{L_m}(\rho)) = 0,
\end{equation}
\begin{equation}
\tr_\mathfrak{we}({\cal D}_{L_m}([\sigma_{ijk}, \rho])) = 0,
\end{equation}
and
\begin{equation}
\tr(\tr_\mathfrak{we}([\sigma_{ijk}, {\cal D}_{L_m}(\rho)])) = 0.
\end{equation}

Notice that the dissipative terms on the environment have no effect on the second derivative
of the purity $\gamma_\mathfrak{l}$. Using the previous expressions we may write the initial
purity loss acceleration as
\begin{equation}
\label{eq:2nd_deriv}
\begin{split}
\ddot\gamma_\mathfrak{l}(\rho)|_{t=0} =&
4 \sum_{\substack{ai > 0 \\ bcjk \geq 0}} g_{abc} g_{ijk} \tr(\sigma_a^\mathfrak{l} \sigma_i^\mathfrak{l} \rho_\mathfrak{l}^2 - \sigma_a^\mathfrak{l} \rho_\mathfrak{l} \sigma_i^\mathfrak{l} \rho_\mathfrak{l}) \\
&\qquad (\tr(\sigma_b^\mathfrak{w} \rho_\mathfrak{w}) \tr(\sigma_j^\mathfrak{w} \rho_\mathfrak{w})
\tr(\sigma_c^\mathfrak{e} \rho_\mathfrak{e}) \tr(\sigma_k^\mathfrak{e} \rho_\mathfrak{e})
- \tr(\sigma_b^\mathfrak{w} \sigma_j^\mathfrak{w} \rho_\mathfrak{w}) \tr(\sigma_c^\mathfrak{e} \sigma_k^\mathfrak{e} \rho_\mathfrak{e})).
\end{split}
\end{equation}

Notice that the terms with $a = 0$ or $i = 0$ do not contribute to this sum,
since the trace $
\tr(\sigma_a^\mathfrak{l} \sigma_i^\mathfrak{l} \rho_\mathfrak{l}^2 - \sigma_a^\mathfrak{l} \rho_\mathfrak{l} \sigma_i^\mathfrak{l} \rho_\mathfrak{l})
$ is zero in such cases, i.e.
the local Hamiltonians $H_\mathfrak{w}$, $H_\mathfrak{e}$ and the
interaction Hamiltonian $H_\mathfrak{we}$ do not contribute to the initial
purity loss acceleration.
Therefore we can remove them from the sum, as done above.
This is also the
case when either $b = 0$ or $j = 0$, and at the same time either $c = 0$ or $k = 0$, i.e. when $b = c = 0$, $b = k = 0$, $j = c = 0$ or $j = k = 0$.
In order to find a function that depends only on $\rho_\mathfrak{w}$, we may consider the initial
state of the environment as part of the model of the system, i.e. we may fix an initial state
$\rho_\mathfrak{e}$ which is assumed to be known. Then, we may average the purity loss
acceleration over all the initial pure states of the information subsystem, i.e. we compute
\begin{equation}
\langle \ddot \gamma_\mathfrak{l}(\rho)|_{t=0} \rangle_\mathfrak{l} =
\frac{1}{\nu_\mathfrak{l}} \int_{{\rm extr} ({\cal D}({\cal H}_\mathfrak{l}))} \ddot \gamma_\mathfrak{l}(\rho)|_{t=0} d\rho_\mathfrak{l},
\end{equation}
where $\nu_\mathfrak{l} = \int_{{\rm extr}({\cal D}({\cal H}_\mathfrak{l}))} d\rho_\mathfrak{l}$.

\noindent The computation of the average consists in solving integrals of the following kind:
\begin{equation}
I_{a,i} = \int_{{\rm extr} ({\cal D}({\cal H}_\mathfrak{l}))}
\tr(\sigma_a^\mathfrak{l} \sigma_i^\mathfrak{l} \rho_\mathfrak{l}^2 -
\sigma_a^\mathfrak{l} \rho_\mathfrak{l} \sigma_i^\mathfrak{l} \rho_\mathfrak{l}) d\rho_\mathfrak{l}.
\end{equation}
By linearity, we can split $I_{a, i}$ into two integrals, such that
\begin{equation}
I_{a,i} = \nu_\mathfrak{l} (\langle \sigma_a^\mathfrak{l} \sigma_i^\mathfrak{l} \rho_\mathfrak{l} \rangle -
\langle \sigma_a^\mathfrak{l} \rho_\mathfrak{l} \sigma_i^\mathfrak{l} \rho_\mathfrak{l} \rangle),
\end{equation}
where we have also used the fact that for any pure state $\rho_\mathfrak{l}^2 = \rho_\mathfrak{l}$.
This integral is then solved by applying the results of the previous lemma.

Finally, let $\tau_{ck} = \tr(\sigma_c^\mathfrak{e} \sigma_k^\mathfrak{e} \rho_\mathfrak{e})$ and
$\nu_n = \frac{1}{n_\mathfrak{l}} - \frac{1}{n_\mathfrak{l} (n_\mathfrak{l} + 1)}$.
Putting all the terms together we obtain,
\begin{equation}
\langle \ddot\gamma_\mathfrak{l}(\rho)|_{t=0} \rangle_\mathfrak{l} =
4 \nu_n \sum_{\substack{i > 0 \\ bcjk \geq 0}} g_{ibc} g_{ijk}
(\tau_{c0} \tau_{0k} \tr(\sigma_b^\mathfrak{w} \rho_\mathfrak{w}) \tr(\sigma_j^\mathfrak{w} \rho_\mathfrak{w})
- \tau_{ck} \tr(\sigma_b^\mathfrak{w} \sigma_j^\mathfrak{w} \rho_\mathfrak{w})).
\end{equation}
\end{proof}

\section{Proof of Proposition \ref{thm:optimwall}}
\label{sec:qubitproof}

We begin by proving some useful properties of the operators $D_i$ in the special case of $n_\mathfrak{w} = 2$.
For larger wall subsystems, these properties do not apply,
but we can still use the eigenstates of $D_1$ as the initial points of the gradient descent algorithm
to solve \eqref{eq:optprob_state_full} or \eqref{eq:optprob_state}.

\begin{lemma}
\label{thm:qubitprops}
Given a set of Hermitian operators $\{D_i\}_i$ such that $\tr(D_i D_j) = \delta_{ij}$ and $\tr(D_i) = 0$,
let $\{\lambda_k^i\}_k$ and $\{\ket{\psi_k^i}\}_k$ be their eigenvalues and eigenvectors, respectively,
i.e. $D_i \ket{\psi_k^i} = \lambda_k^i \ket{\psi_k^i}$. If they act on qubits, the following properties are
satisfied:

\begin{enumerate}
\item $\lambda_1^i = - \lambda_2^i$
\item $\bra{\psi_k^i} D_j \ket{\psi_k^i} = \lambda_k^j \delta_{ij}$
\item $D_i^2 = \frac{1}{2} \one$
\item $(\lambda_k^i)^2 = \frac{1}{2}$
\end{enumerate}

\end{lemma}
\begin{proof}
Property (1) is a direct consequence of $\tr(D_i) = \lambda_1^i + \lambda_2^i = 0$.
Property (2) is also obtained by direct calculation. Assume $i \neq j$, then
\begin{equation*}
\begin{split}
\tr(D_i D_j) =&
\tr(D_j \lambda_1^i \ketbra{\psi_1^i}{\psi_1^i} + D_j \lambda_2^i \ketbra{\psi_2^i}{\psi_2^i}) \\
=& \lambda_1^i \bra{\psi_1^i} D_j \ket{\psi_1^i} + \lambda_2^i \bra{\psi_2^i} D_j \ket{\psi_2^i} \\
=& \lambda_1^i (\bra{\psi_1^i} D_j \ket{\psi_1^i} - \bra{\psi_2^i} D_j \ket{\psi_2^i}) = 0.
\end{split}
\end{equation*}
Therefore, $\bra{\psi_1^i} D_j \ket{\psi_1^i} = \bra{\psi_2^i} D_j \ket{\psi_2^i}$. But we also
have that $\tr(D_i) = \bra{\psi_1^i} D_j \ket{\psi_1^i} + \bra{\psi_2^i} D_j \ket{\psi_2^i} = 0$.
This implies that $\bra{\psi_1^i} D_j \ket{\psi_1^i} = \bra{\psi_2^i} D_j \ket{\psi_2^i} = 0$ whenever
$i \neq j$. On the other hand, if $i = j$, we have that $\bra{\psi_k^i} D_i \ket{\psi_k^i} = \lambda_k^i$.
Putting the two cases together proves property (2).

Finally, property (3) is a consequence of $D_i = D_i^\dagger$, property (1) and $\tr(D_i^2) = 1$.

\end{proof}

With these properties in hand, we have the tools necessary to prove Proposition \ref{thm:optimwall}.

\begin{proposition*}
    Given $H_\mathfrak{lw}$ in the OSD form as in (\ref{eq:osd_ham}), $n_\mathfrak{w} = 2$
    and $s_1 > s_i$ for $i > 1$, choosing $\ket{w} \in \mathrm{eigenvectors}(D_1)$
    minimizes $\Gamma_2(\ket{w})$ over $\bigcup_{i = 1}^3 \mathrm{eigenvectors}(D_i)$,
    and locally minimizes $\Gamma_2(\ket{w})$ over the set of pure states of ${\cal H}_\mathfrak{w}$
\end{proposition*}
\begin{proof}

The Riemannian gradient \eqref{eq:osd_gamma_grad} of $\Gamma_2(\ket{w})$ is equal to zero if and only if
\begin{equation}
\begin{split}
\sum_i s^2_i (D_i^2
- 2 \bra{w} D_i \ket{w} D_i) \ket{w} =
\sum_i s^2_i (\bra{w} D_i^2 \ket{w}
- 2 (\bra{w} D_i \ket{w})^2) \ket{w},
\end{split}
\end{equation}
i.e. if $\ket{w}$ is an eigenstate of $\sum_i s^2_i (D_i^2 - 2 \bra{w} D_i \ket{w} D_i)$.

Since $n_\mathfrak{w} = 2$, i.e. ${\cal H}_\mathfrak{w}$ is a qubit system, we can use the properties given in Lemma \ref{thm:qubitprops}.
By applying property 3 of lemma \ref{thm:qubitprops}, the previous expression simplifies to
\begin{equation}
\sum_{i=1}^3 s^2_i \bra{w} D_i \ket{w} D_i \ket{w} =
\sum_{i=1}^3 s^2_i (\bra{w} D_i \ket{w})^2 \ket{w},
\end{equation}
i.e. $\ket{w}$ must be an eigenstate of $\sum_i s^2_i \bra{w} D_i \ket{w} D_i$. Consider now property
2 of lemma \ref{thm:qubitprops} and $\ket{w} = \ket{\psi_k^j}$. We get that
$\sum_i s^2_i \bra{w} D_i \ket{w} D_i = s^2_j \lambda_k^j D_j$, i.e. the eigenvectors
of each of the operators $D_i$ are critical points of the function $\Gamma_2(\rho_\mathfrak{w})$ for qubit walls.

In order to test if these critical points are minima of $\Gamma_2(\rho_\mathfrak{w})$, we can
compute the Riemannian Hessian at these points and check that it is positive definite.
The Euclidean Hessian is given by
\begin{equation}
\begin{split}
\nabla^2 \Gamma_2(\rho_\mathfrak{w}) =& \sum_i s^2_i (4 D_i^2 - 8 \bra{w} D_i \ket{w} D_i
- 8 D_i \ketbra{w}{w} D_i).
\end{split}
\end{equation}
The Riemannian Hessian can be computed as
\[
{\rm Hess } \,\Gamma_2(\rho_\mathfrak{w}) [\vec{\epsilon}] =
(\one - \ketbra{w}{w}) \nabla^2 \Gamma_2(\rho_\mathfrak{w}) \vec{\epsilon} -
\vec{\epsilon} \bra{w} (\ketbra{w}{w} \nabla \Gamma_2(\rho_\mathfrak{w})),
\] where $\vec{\epsilon} \in {\cal T}_{\ket{w}} {\cal M}$, resulting in
\begin{equation}
\begin{split}
{\rm Hess}\, \Gamma_2(\rho_\mathfrak{w}) [\vec{\epsilon}] =&
(\one - \ketbra{w}{w}) \\
& (\sum_i s^2_i (4 D_i^2 - 8 \bra{w} D_i \ket{w} D_i
- 8 D_i \ketbra{w}{w} D_i)) \vec{\epsilon} \\
& - \sum_i s^2_i (2 \bra{w} D_i^2 \ket{w} - 4 (\bra{w} D_i \ket{w})^2) \vec{\epsilon}.
\end{split}
\end{equation}

By imposing that the wall be a qubit system and $\ket{w}$ be an eigenstate
of $D_1$, the Hessian simplifies to
\begin{equation}
{\rm Hess} \Gamma_2(\rho_\mathfrak{w}) [\vec{\epsilon}] = (10 s^2_1 - 3 \sum_{i=1}^3 s^2_i) \vec{\epsilon}
\geq s^2_1 \vec{\epsilon}.
\end{equation}
Since $s^2_1 > 0$, we have that taking $\ket{w}$ as an eigenstate of $D_1$ gives a local
minimum of $\Gamma_2(\rho_\mathfrak{w})$.

Consider now equation (\ref{eq:osd_gamma}), applying properties 2-4 of lemma \ref{thm:qubitprops} and
taking $\ket{w} = \ket{\psi_k^j}$, we get obtain
\begin{equation}
\Gamma_2(\ketbra{\psi_k^j}{\psi_k^j}) = 1 - \frac{s^2_j}{2},
\end{equation}
which is minimized over all $j$ when $j = 1$.
\end{proof}

\clearpage

\bibliographystyle{iopart-num}
\bibliography{ref}

\end{document}